\documentclass[11pt]{article}
\usepackage{amssymb,amsmath,amsthm,extpfeil,amsfonts,mathrsfs}
\usepackage{color}
\usepackage{subfigure}
\allowdisplaybreaks[4]
\usepackage{accents}
\usepackage{threeparttable}
\usepackage{caption}
\usepackage{mathrsfs}
\usepackage{enumerate}
\usepackage{bm}
\usepackage{bbm}
\usepackage{graphicx}       
\usepackage{epstopdf}
\usepackage{booktabs}
\usepackage{float}
\usepackage{longtable}
\usepackage{verbatim}
\usepackage{multirow}
\usepackage{geometry}
\usepackage{caption}
\newcommand{\RNum}[1]{\uppercase\expandafter{\romannumeral #1\relax}}
\numberwithin{equation}{section}
\newtheorem{theorem}{Theorem}[section]
\newtheorem{lemma}{Lemma}[section]

\renewcommand{\thefootnote}{\fnsymbol{footnote}}
\newcommand{\dif}{{\mathrm d}}
\newcommand{\me}{{\mathrm e}}

\date{}

\usepackage{color}
\usepackage{hyperref}
\hypersetup{
	colorlinks = true,%
	citecolor = blue,
	filecolor=red,%
	linkcolor = [rgb]{0.65,0.0,0.0},%
	anchorcolor = red,
	pagecolor = red,
	urlcolor= [rgb]{0.65,0.0,0.0}
	linktocpage=true,
	pdfpagelabels=true,
	bookmarksnumbered=true,
}
\begin{document}

\title{\bf{Epidemic waves for a two-group SIRS model with  double nonlocal effects  in a patchy environment}\footnotemark[1]}

\author{Chufen Wu$^{1}$, Yonghui Xia$^{1}$\footnotemark[2]  and  Jianshe Yu$^{2}$ \\
\small $^{1}$ School of Mathematics, Foshan University, Foshan 528000, PR China\\
	\small  $^2$ College of Mathematics and Information Sciences, Guangzhou University, Guangzhou, 510006, China}
\renewcommand{\thefootnote}{\fnsymbol{footnote}}

\footnotetext[1]{Research supported by  National Natural Science Foundation of China (No. 12071074) and
 Guangdong Basic and Applied Basic Research Foundation (No. 2023A1515012078).}
\footnotetext[2]{{\small Corresponding author.\\ $~~~~~$ E-mail addresses: wucfmath@fosu.edu.cn,yhxia@zjun.cn,
xiayonghui@fosu.edu.cn, jsyu@gzhu.edu.cn.}}
\date{}

%
%

\maketitle
\begin{center}
\begin{minipage}{13cm}
\par
\small  {\bf Abstract:} We propose a lattice dynamical system  that arises in  a discrete diffusive two-group epidemic model with latency in a patchy environment. The model considers the SIS form and  latency of the disease  in group 1, while the SIR form without latency of the disease  in group 2. The system includes double nonlocal effects, one effect is the nonlocal diffusion of individuals in isolated patches or niches, while the other effect is the distributed transmission delay representing the incubation of the disease. We  demonstrate  that there is a threshold  value $c^*$ that can determine the persistence or disappearance of the disease. If  $c\geq c^*$, then there is an epidemic wave connecting the disease-free equilibrium  and endemic equilibrium. In this case, the disease will evolve to endemic.  If $0<c<c^*$,  then the disease will die out.


\end{minipage}
\end{center}

%
%

{\section{Introduction}}

Mathematical modeling and analysis play a crucial role in disease prevention, transmission, and control. Through mathematical modeling, we can make better sense of the laws of disease transmission and  predict the trend or speed of disease propagation \cite{Hethcote2}.
Thus we could provide some insights on  the formulation of public health policies and  effectively reduce the spread of diseases.
One of the key topics is  to find out a threshold that decides the persistence or
extinction of  diseases \cite{am}.
Spatial heterogeneity is  significant for the persistence and dynamics of epidemics since  asynchrony between populations among dissimilar areas  has enabled  global persistence, even when the disease disappears locally \cite{lm}.  Patch models, which measure disease transmission, do not assume homogeneous mixing of members. For the formulation and analysis of some patch
models such as the review articles \cite{av,w} and  references therein, have been studied. Due to the vital impact of host heterogeneity on the dynamics of infectious diseases, multi-group epidemic models have received noticeable attention in recent years  \cite{deb}. Models with different groups are more accurate and essential, for example, hand-foot-mouth disease, chickenpox, scarlet fever typically affect children while sexually transmitted diseases are more common among adults.
To some extent, changing a model by joining multi-groups together can change the asymptotic behaviors of the infectious diseases  \cite{Hethcote1}.
To test this, Hethcote \cite{Hethcote1} proposed an SIS epidemic model with two different interacting groups between two patches and studied the
asymptotic stability of equilibria.
Lloyd and  May \cite{lm} divided the population into $n$ subpopulations and introduced a  multi-patch (metapopulation) model for spatial
heterogeneity in epidemics. They have demonstrated that patches in non-seasonal deterministic models always oscillate in phase
while the weakest between patch coupling.
 Wang and Zhao \cite{wz}  considered  a disease transmission model
with population dispersal among $n$ patches  is associated with the threshold for  the persistence or vanishing of the disease.
Brauer, Van den Driessche and  Wang \cite{bvw} have examined an SIRS model in one and two patches,
observing that oscillations might happen within one patch and that travel between patches may trigger
essential and unexpected changes in the behavior of the system. For epidemic models on metapopulation, we also refer the reader to
Castillo-Chavez and Yakubu \cite{cy}, Brauer and van den Driessche \cite{bv}, Wang and Mulone \cite{wm},  Hsieh,  van den Driessche  and Wang \cite{hvw}  and the references therein.

Due to the incubation period and  time lag of transmission of infectious diseases, time delays also need to be included in the model.
The presence of time delay can slow down the spread of diseases, as infected individuals require time to transmit the disease to others. Time delay may also lead to fluctuations in the spread of  diseases, because delayed transmission by infected individuals can result in both outbreaks and declines of  diseases. Salmani and van den Driessche \cite{sv}  studied an  epidemic model with a latent period for $n$ patches to
characterize the dynamics of an infectious disease in a population in which individuals
migration between patches.  Li and Zou \cite{lz}  formulated an SIR model  for the population living in  $n$ patches and investigated the persistence and extinction of the model system with a fixed delay.

Lattice dynamical systems in  epidemiology can model the evolution of susceptible and infected individuals in discrete niches or patchy environments.
Based on  this consideration, we analyze a lattice dynamical system that emerges in a discrete diffusive two-group epidemic model with latency as follows:
\begin{equation}\label{1.2}
\left\{\begin{array}{l}
\frac{\dif }{\dif t}S_j(t)=d_1\mathcal {A}[S_j](t)+b_1[K_1-S_{j}(t)]- S_{j}(t)I_{j}(t)+\sigma I_{j}(t), \\
\frac{\dif }{\dif t}P_j(t)=d_2\mathcal {A}[P_j](t)+b_2[K_2-P_{j}(t)]- P_{j}(t)I_{j}(t), \\
\frac{\dif }{\dif t}I_{j}(t)=d_3\mathcal {A}[I_j](t)+\gamma_1\displaystyle{\int_{0}^h f(\tau)S_{j}(t-\tau)I_{j}(t-\tau)\dif\tau}-\delta I_{j}(t)+\gamma_2P_{j}(t)I_{j}(t),
\end{array}\right.
\end{equation}
where $j\in\mathbb{Z}$, $\mathcal {A}[u_j](t):=\sum\limits_{i\neq0}J(i)[u_{j-i}(t)-u_j(t)]$ and $d_1, d_2, d_3$ are the dispersal rates of each compartments.   In figure 1, $S_j(t), P_j(t)$
denote the densities of  susceptible individuals of group $1$ and group $2$ at niches $j$ and time $t>0$, respectively.  The densities of  infected individuals  at niches $j$ and time $t$ are represented by $I_j(t)$; $\sigma$ is the relapse rate, $b_1[K_1-S_{j}(t)]$ and $b_2[K_2-P_{j}(t)]$ denote recruitments or external supplies of group 1 and group 2. The constant $\delta$ is a composite rate including the death rate and removal rate of infected individuals. Let $h$ be a superior limit of incubation and $f(\tau)$ be the fraction of vector population in which the time taken to become infectious is the delay $\tau$. Consider the bilinear incidence rate or mass action form, which  is a commonly used incidence for both human and
animal diseases.

\begin{figure}[H]
\centering
\includegraphics[height=7.0cm,width=9.0cm]{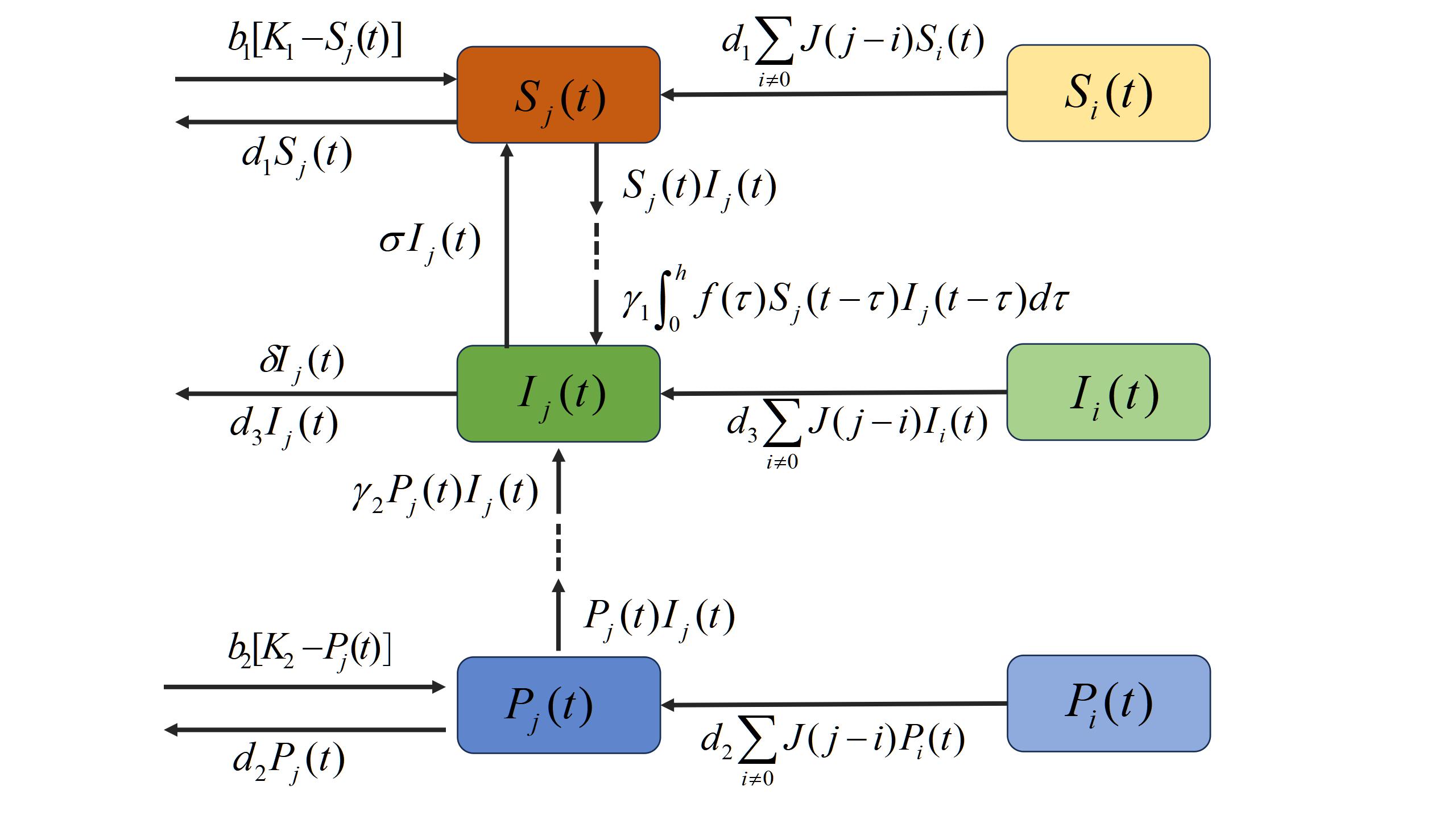}
\caption{The flow diagram of  the disease transmission.}
\end{figure}

To determine if a disease can spread at a consistent rate, we always  investigate the phenomenon known as traveling wave solutions (referred to
as epidemic waves). As epidemic waves come and go, they can indicate whether a disease is endemic or has disappeared.
The connection of epidemic waves in two tails shows the evolution of the disease over time.
In recent years, Fu, Guo and Wu \cite{fgw} investigated the existence and non-existence of epidemic waves
for the following  discrete diffusive SI model:
 \begin{equation}\label{1.3}
\left\{\begin{array}{l}
\frac{\dif }{\dif t}S_j(t)=[S_{j+1}(t)+S_{j-1}(t)-2S_{j}(t)]-\gamma S_{j}(t)I_{j}(t), \\
\frac{\dif }{\dif t}I_{j}(t)=d[I_{j+1}(t)+I_{j-1}(t)-2I_{j}(t)]+\gamma S_{j}(t)I_{j}(t)-\delta I_{j}(t).
\end{array}\right.
\end{equation}
In model system  \eqref{1.3}, there is no constant recruitment of susceptibles. In this case, the epidemic wave is a mix of front type ($S$-component) and pulse type ($I$-component). Further, Wu \cite{wu} supplemented the existence of epidemic waves with the critical speed for  \eqref{1.3}. Later, Chen, Guo and Hamel \cite{cgh} established  the existence and non-existence of epidemic waves  by introducing the  constant recruitment in \eqref{1.3}, that is,
 \begin{equation*}\label{1.4}
\left\{\begin{array}{l}
\frac{\dif }{\dif t}S_j(t)=[S_{j+1}(t)+S_{j-1}(t)-2S_{j}(t)]+\mu[1-S_{j}(t)]-\gamma S_{j}(t)I_{j}(t), \\
\frac{\dif }{\dif t}I_{j}(t)=d[I_{j+1}(t)+I_{j-1}(t)-2I_{j}(t)]+\gamma S_{j}(t)I_{j}(t)-(\mu+\delta) I_{j}(t).
\end{array}\right.
\end{equation*}
In this scenario, the epidemic wave is the front type of each component. Subsequently, Zhang, Wang and Liu \cite{zwl}
proved  the epidemic waves  would approach to the endemic equilibrium as $j+ct\to \infty$ through the creation of a Lyapunov functional.
When considering the combined effects of multi-group interactions, latent periods, and patchy environments, the dynamics of model systems become more interesting and complicated.
San et al.  \cite{sw,sh} obtained the existence and non-existence  of epidemic waves for a discrete  two-group SI model with
latent delay and bilinear incidence  without recruitment of susceptibles.  Zhou, Song and Wei \cite{zsw} derived the existence and non-existence of mixed types of epidemic waves in a discrete diffusive SI model with saturated  incidence and discrete time delay.
Shortly before, Li and Zhang \cite{lzhang} studied the boundedness and persistence of epidemic waves in a discrete diffusive SI model with general incidence and distributed  time delays. Inspired by aforementioned  works,  we consider the model system \eqref{1.2} and establish the existence of  bounded super-critical and critical epidemic  waves.  We also demonstrate the  non-existence of non-trivial, positive and bounded epidemic  waves in the sub-critical situation based on the threshold value range.

Assume all the parameters  in \eqref{1.2} are positive. Further  assume
\begin{itemize}
\item[(J)]\;$J(i)=J(-i)\geq0$, $\sum_{i\neq0}J(i)=1$ and
 the support of $J(i)$ is  bounded from above and below, i.e.,
there exists an $N_1\in\mathbb{Z}_+$ such that $J(i)=0$ for  $|i|> N_1$;
\item[(f)]\;$f(\cdot)\in C([0,h],\mathbb{R}_+)$,  $\int_{0}^hf(\tau)\dif \tau=1$;
\item[(H)] $\sigma<\frac{\delta}{\gamma_1}<K_1$.
\end{itemize}
The model system \eqref{1.2}  always admits a disease-free equilibrium $(K_1,K_2,0)$.
Denote
\begin{equation*}
e(I)=\gamma_1\frac{b_1K_1+\sigma I}{b_1+I}+\gamma_2\frac{b_2K_2}{b_2+I}-\delta.
\end{equation*}
Let  (H) hold.  There exists a unique  $I_*>0$ such that $e(I_*)=0$. Denote $S_*=\frac{b_1K_1+\sigma I_*}{b_1+I_*}$ and $N_*=\frac{b_2K_2}{b_2+I_*}$.
Then \eqref{1.2} has a unique positive equilibrium $(S_*,P_*,I_*)$  when (H) holds.

Epidemic model systems like \eqref{1.2}  generally do not generate monotone dynamical or semidynamical
systems, which result in the tools of monotone dynamical system theory  are
not directly usable for their analyses. Compared to distinct progress in monotone  lattice dynamical systems  with monostable nonlinearity (e.g, see \cite{zhh,hz,cg,fwz1,chen,mz,cfg,gh,gw,xm}) and bistable nonlinearity (e.g, see \cite{k,z,hz1,cf,cpw,cmv,chengw,hhv,fz}),
there are very few studies on non-monotone  lattice dynamical systems. The existence of epidemic waves of \eqref{1.2} is more difficult to prove when there are recruitments. In this situation, each component of \eqref{1.2} may oscillate and the connection of epidemic waves in one tail is the endemic equilibrium. Consequently, we need to demonstrate that each component of \eqref{1.2} exhibits strong persistence and exponential decay at infinity.
Then we can construct a Lyapunov functional combined with Lasalle's  invariance principle to show the epidemic waves converge to the endemic equilibrium. However, the construction of Lyapunov functional here  is non-trivial due to the double nonlocal effects and time delay. Motivated by the work of Huang and Wu \cite{hw},  Li and Zhang \cite{lzhang}  as well as Zhang et al. \cite{zwl}, we introduce certain  constituents of the Lyapunov functional to counteract the delay and nonlocal effects. The bilinear incidence rate, which is  unsaturated and leads to the $I$-component of \eqref{1.2} is unbounded. The boundedness and strong persistence of the $I$-component in \eqref{1.2} are hard to prove due to the double effects of nonlocality. By first making suitable shifts of the three components of \eqref{1.2} and utilizing the local convergence of translated variables,  we are able to pull out the $S$-variable from the convolution. Then the problem reduces to  the positive solution of a linear and nonlocal characteristic equation with distributed  delay.
Furthermore, we consider a decreasing sequence of super-critical speeds and pass to the limit so as to obtain the existence of critical waves.
However, the limit maybe trivial given that the model system \eqref{1.2} is nonmonotone. Therefore, a thorough analysis is required to address this issue.
At last, the non-existence of epidemic waves  is derived by estimating the exponential decay behaviors of $K_1-S(\cdot), K_2-P(\cdot)$ and $I(\cdot)$ at infinity and
using the singularity analysis of two-sides Laplace transform.

The organization of this paper is as follows.
In Sections 2 and 3, we prove the existence of epidemic waves  connecting the disease-free equilibrium $(K_1,K_2,0)$ and endemic equilibrium $(S_*,P_*,I_*)$ for $c\geq c^*$. In the first step, by constructing upper and lower solutions in bounded domains and passing to the limit in the whole line, we get the semi-waves. Next, through careful analysis and the construction of a Lyapunov functional, we  determine the complete waves. In particular, in the critical case, prior estimates are performed and it is shown that the solution sequences are uniformly bounded. Through proper shifts, negative Laplace transforms, and the fluctuation lemma, we demonstrate that the boundary conditions are satisfied at minus infinity.
 Section 4 is devoted to showing the non-existence of epidemic waves for  $0<c<c^*$.
In the last section, we implemented some numerical simulations to study the qualitative properties of epidemic waves.

%
%

{\section{Existence of  epidemic wave for $c>c^*$}}
\setcounter{equation}{0}

Letting $(S_{j}(t),P_{j}(t),I_{j}(t))=(S(\xi),P(\xi),I(\xi))$ with $\xi=j+ct$ and plugging them into \eqref{1.2}, we   have
\begin{equation}\label{2.1}
\left\{\begin{array}{l}
cS'(\xi)=d_1\mathcal {A}[S](\xi)+ b_1[K_1-S(\xi)]- S(\xi)I(\xi)+\sigma I(\xi),\\
cP'(\xi)=d_2\mathcal {A}[P](\xi)+ b_2[K_2-P(\xi)]- P(\xi)I(\xi),\\
cI'(\xi)=d_3\mathcal {A}[I](\xi)+\gamma_1 \displaystyle{\int_{0}^hf(\tau)S(\xi-c\tau)I(\xi-c\tau)\dif \tau}-\delta I(\xi)+\gamma_2P(\xi)I(\xi),
\end{array}\right.
\end{equation}
where $\mathcal {A}[u](\xi):=\sum_{i\neq0}J(i)[u(\xi-i)-u(\xi)]$. Based on the assumption (J), then
\begin{equation*}
\mathcal {A}[u](\xi)
=\sum_{|i|=1}^{N_1}J(i)[u(\xi-i)-u(\xi)]
=\sum_{i=1}^{N_1}J(i)[u(\xi-i)+u(\xi+i)-2u(\xi)].
\end{equation*}
We assume the solution $(S,P,I)$ of \eqref{2.1} satisfies the following  boundary conditions:
\begin{equation}\label{2.2a}
\lim_{\xi\rightarrow-\infty}(S,P,I)(\xi)=(K_1,K_2,0),\quad \lim_{\xi\rightarrow\infty}(S,P,I)(\xi)=(S_*,P_*,I_*).
\end{equation}
For complex number $\lambda$ and positive constant $c$, we define
\begin{equation}\label{2.3}
\triangle_K(\lambda,c)=d_3\sum_{i=1}^{N_1}J(i)[\me^{-\lambda i}+\me^{\lambda i}]-c\lambda+\gamma_1 K_1\int_{0}^h f(\tau)\me^{-\lambda c\tau}\dif \tau+\gamma_2K_2-(d_3+\delta).
\end{equation}
Then the  following result is straightforward.
\begin{lemma}\label{L1}
Assume that (J), (f) and (H) hold, then there exists a pair of
$(\lambda^*,c^*)$ with $\lambda^*,c^*>0$ such that
\begin{itemize}
\item[(i)] $\triangle_{K}(\lambda^*,c^*)=0,
\partial_{\lambda}\triangle_{K}(\lambda^*,c^*)=0$;
\item[(ii)] $\triangle_{K}(\lambda,c)>0$ for $0<c<c^*$ and
$\lambda\in\mathbb{R}$;
\item[(iii)] $\triangle_{K}(\lambda,c)=0$ has two zeros
$0<\lambda_1(c)<\lambda_2(c)<\infty$ for $c>c^*$,   $\triangle_{K}(\lambda,c)>0$ for $\lambda\in(0,\lambda_1)\cup (\lambda_2,\infty)$ and  $\triangle_{K}(\lambda,c)<0$ for $\lambda\in(\lambda_1,\lambda_2)$.
\end{itemize}
\end{lemma}
Suppose $c>c^*$. Define six non-negative continuous functions as follows
\begin{equation*}
\begin{aligned}
&S_{+}(\xi)\equiv K_1,\quad \qquad\qquad \qquad\quad\; S_{-}(\xi)=\max\left\{K_1-\alpha_1 \me^{\beta_1\xi},\sigma\right\}, \\
&P_{+}(\xi)\equiv K_2, \qquad \qquad\qquad \qquad\; P_{-}(\xi)=\max\left\{K_2-\alpha_2 \me^{\beta_2\xi},0\right\}, \\
&I_{+}(\xi)=\me^{\lambda_1\xi},\qquad\qquad \qquad\qquad\; I_{-}(\xi)=\max\left\{\me^{\lambda_1\xi}(1-L\me^{\eta\xi}),0\right\},\\
\end{aligned}
\end{equation*}
where  $\alpha_1, \alpha_2, \beta_1, \beta_2, L, \eta$ are all positive constants to be determined later.
\begin{lemma}\label{L2}
The continuous functions $(S_+,P_+,I_+)$ and $(S_-,P_-,I_-)$ are a pair of
upper and lower solutions of \eqref{2.1} for  $c>c^*$ in the sense that
\begin{equation}\label{2.4}
\begin{aligned}
&cS_+'(\xi)\geq d_1\mathcal {A}[S_+](\xi)+ b_1[K_1-S_+(\xi)]- [S_+(\xi)-\sigma]I_-(\xi),\\
&cP_+'(\xi)\geq d_2\mathcal {A}[P_+](\xi)+b_2[K_2-P_+(\xi)]-P_+(\xi)I_-(\xi), \\
&cI_+'(\xi)\geq d_3\mathcal {A}[I_+](\xi)+\gamma_1 \displaystyle{\int_{0}^hf(\tau)S_+(\xi-c\tau)I_+(\xi-c\tau)\dif \tau}-\delta I_+(\xi)+\gamma_2P_+(\xi)I_+(\xi),
\end{aligned}
\end{equation}
for any  $\xi\in \mathbb{R}$ and
\begin{equation*}
\begin{aligned}
&cS_-'(\xi)\leq d_1\mathcal {A}[S_-](\xi)+ b_1[K_1-S_-(\xi)]- [S_-(\xi)-\sigma] I_+(\xi), \quad  \forall\;\xi\neq \xi_1;\\
&cP_-'(\xi)\leq d_2\mathcal {A}[P_-](\xi)+ b_2[K_2-P_-(\xi)]-P_-(\xi)I_+(\xi) , \quad \forall\; \xi\neq \xi_2;\\
&cI_-'(\xi)\leq d_3\mathcal {A}[I_-](\xi)+\gamma_1 \displaystyle{\int_{0}^hf(\tau)S_-(\xi-c\tau)I_-(\xi-c\tau)\dif \tau}-\delta I_-(\xi)+\gamma _2P_-(\xi)I_-(\xi),\; \forall\;  \xi\neq \xi_3,
\end{aligned}
\end{equation*}
where $\xi_1:=-\frac{1}{\beta_1}\ln \frac{\alpha_1}{K_1 -\sigma}, \xi_2:=-\frac{1}{\beta_2}\ln \frac{\alpha_2}{K_2}$ and $\xi_3:=-\frac{1}{\eta}\ln L$.
\end{lemma}
\begin{proof}
It is easy to see that  \eqref{2.4} is satisfied.

Let $\beta_1<\lambda_1$ be sufficiently small and $\alpha_1>K_1-\sigma$ be sufficiently large.  Then  $\xi_1<0$ is sufficiently small.
If $\xi<\xi_1$, then $S_{-}(\xi)=K_1-\alpha_1 \me^{\beta_1\xi}>\sigma$.
Noting that $I_{+}(\xi)=\me^{\lambda_1\xi}$  for all $\xi\in\mathbb{R}$,  we then follows   that
\begin{equation*}
\begin{aligned}
& d_1\mathcal {A}[S_-](\xi)-cS_-'(\xi)+ b_1[K_1-S_-(\xi)]-[S_-(\xi)-\sigma]I_+(\xi) \\
\geq &d_1\alpha_1 \me^{\beta_1\xi}\sum_{i\neq0}J(i)(1-\me^{-\beta_1 i})+c\alpha_1\beta_1\me^{\beta_1\xi}
+b_1\alpha_1\me^{\beta_1\xi}-[K_1-\sigma-\alpha_1 \me^{\beta_1\xi}]\me^{\lambda_1\xi}\\
= &\alpha_1\me^{\beta_1\xi}\Big[d_1\sum_{i\neq0}J(i)(1-\me^{-\beta_1 i})+c\beta_1+b_1
-\frac{{K}_1-\sigma}{\alpha_1}\me^{(\lambda_1-\beta_1)\xi}+\me^{\lambda_1\xi}\Big]\\
\geq &\alpha_1\me^{\beta_1\xi}\Big[d_1\sum_{i\neq0}J(i)(1-\me^{-\beta_1 i})+c\beta_1+
b_1-\frac{{K}_1-\sigma}{\alpha_1}+\me^{\lambda_1\xi}\Big]\geq 0,
\end{aligned}
\end{equation*}
due to $\beta_1$   sufficiently small and $\alpha_1$ large.
If $\xi>\xi_1$, then $S_{-}(\xi)=\sigma$ and the above inequality holds obviously.
Pick $\beta_2<\lambda_1$  and $\alpha_2>K_2$. In an analogous way, we can prove
$$d_2\mathcal {A}[P_-](\xi)-cP_-'(\xi)+ b_2[K_2-P_-(\xi)]-P_-(\xi)I_+(\xi) \geq 0   \quad\text{ for }\xi\neq \xi_2.$$

Let $0<\eta<\min\{\beta_1,\beta_2,\lambda_2-\lambda_1\}$.
Choose
$$L\geq \max\left\{\frac{\alpha_1}{{K}_1-\sigma},\frac{\alpha_2}{K_2},\frac{\alpha_1\gamma_1+\alpha_2\gamma_2}{-\triangle_{K}(\lambda_1+\eta,c)}\right\}$$
such that $\xi_3<\xi_1$ and $\xi_3<\xi_2$. If $\xi<\xi_3$, then $I_{-}(\xi)=\me^{\lambda_1\xi}(1-L\me^{\eta\xi})$,
$P_{-}(\xi)=K_2-\alpha_2 \me^{\beta_2\xi}$,  $I_{-}(\xi-c\tau)=\me^{\lambda_1(\xi-c\tau)}[1-L\me^{\eta(\xi-c\tau)}]$ and   $S_{-}(\xi-c\tau)={K}_1-\alpha_1 \me^{\beta_1(\xi-c\tau)}$. Hence
\begin{equation*}
\begin{aligned}
&d_3\mathcal {A}[I_-](\xi)-cI_-'(\xi)+\gamma_1 \displaystyle{\int_{0}^hf(\tau)S_-(\xi-c\tau)I_-(\xi-c\tau)\dif \tau}-\delta I_-(\xi)+\gamma_2P_-(\xi)I_-(\xi)\\
 &\geq \me^{\lambda_1\xi}\triangle_{K}(\lambda_1,c)-L\me^{(\lambda_1+\eta)\xi}\triangle_{K}(\lambda_1+\eta,c)-\alpha_1\gamma_1\displaystyle{\int_{0}^h f(\tau) \me^{(\beta_1+\lambda_1)(\xi-c\tau)}\dif\tau}-\alpha_2\gamma_2\me^{(\beta_2+\lambda_1)\xi}\\
&=\me^{(\lambda_1+\eta)\xi}[-L\triangle_{K}(\lambda_1+\eta,c)-\alpha_1\gamma_1\displaystyle{\int_{0}^h f(\tau)\me^{-(\beta_1+\lambda_1)c\tau} \me^{(\beta_1-\eta)\xi}\dif\tau}-\alpha_2\gamma_2\me^{(\beta_2-\eta)\xi}]\\
&\geq \me^{(\lambda_1+\eta)\xi}[-L\triangle_{K}(\lambda_1+\eta,c)-\alpha_1\gamma_1-\alpha_2\gamma_2]\geq 0.
\end{aligned}
\end{equation*}
If $\xi>\xi_3$, then $I_{-}(\xi)=0$ and thus the inequality holds clearly.
\end{proof}
The upper solution $I_{+}(\xi)$ is unbounded.  To avoid the solution variable $I$ is unbounded, we next consider
the truncated problem as follows:
\begin{equation}\label{2.5}
\left\{\begin{array}{l}
cS'(\xi)=d_1\mathcal {A}[S](\xi)+ b_1[K_1-S(\xi)]-[S(\xi)-\sigma]I(\xi), \quad \xi\in[-k,k],\\
cP'(\xi)=d_2\mathcal {A}[P](\xi)+ b_2[K_2-P(\xi)]- P(\xi)I(\xi),\quad\quad \xi\in[-k,k],\\
cI'(\xi)=d_3\mathcal {A}[I](\xi)+\gamma_1 \displaystyle{\int_{0}^hf(\tau)S(\xi-c\tau)I(\xi-c\tau)\dif \tau}-\delta I(\xi)+\gamma_2P(\xi)I(\xi),\quad \xi\in[-k,k],
\end{array}\right.
\end{equation}
wherein $k>-\xi_3$ and
$$u'(-k)=\lim_{z\searrow0}\frac{u(-k+z)-u(-k)}{z},\quad  u'(k)=\lim_{z\searrow0}\frac{u(k)-u(k-z)}{z}. $$
 Meanwhile we consider the following boundary conditions:
\begin{equation}\label{2.6}
(S(\xi),P(\xi),I(\xi))=(K_1,K_2,I_{+}(\xi)),\; \xi<-k,\qquad (S(\xi),P(\xi),I(\xi))=(S(k),P(k),I(k)),\; \xi>k.
\end{equation}
We now introduce some functional spaces for  convenience. Let $\mathbb{X}_k:= [C([-k,k])]^3$ and
 define
$$\Gamma_k:=\left\{(\psi,\phi,\varphi)\in \mathbb{X}_k:\begin{array}{l}
\mbox{\rm{(i)}} \ (\psi,\phi,\varphi)(-k)=(S_+,P_+,I_+)(-k);\\
\mbox{\rm{(ii)}}\ S_-(\xi)\leq \psi(\xi)\leq S_+(\xi), \\
\quad \; P_-(\xi)\leq \phi(\xi)\leq P_+(\xi),\\
\quad \;I_-(\xi)\leq \varphi(\xi)\leq I_+(\xi), \forall\;\xi\in[-k,k]\end{array}\right\}.$$
For any $(\psi,\phi,\varphi)\in \mathbb{X}_k$, we take the norm as
 $\|(\psi,\phi,\varphi)\|_{\mathbb{X}_k}=\max\{|\psi|_1,|\phi|_1,|\varphi|_1\}$. Here $|\cdot|_1$ is the maximum norm in $C([-k,k])$.
 Then $\Gamma_1$ is a non-empty bounded set in $(\mathbb{X}_k,\|\cdot\|_{\mathbb{X}_k})$.
 For any $(S,P,I)\in \Gamma_k$, we can extend $(S,P,I)$ to be continuous outside the interval $[-k,k]$  as definitions  on \eqref{2.6}.
 Choose $\mu>\max\{d_1+b_1+\me^{\lambda_1 k},d_2+b_2+\me^{\lambda_1 k},\delta+d_3\}$ and denote
 \begin{equation*}
 \begin{aligned}
\mathcal{P}_1[S,P,I](\xi)&=(\mu -b_1)S(\xi)+d_1\mathcal {A}[S](\xi)+b_1K_1-[S(\xi)-\sigma]I(\xi),\\
\mathcal{P}_2[S,P,I](\xi)&=(\mu -b_2) P(\xi)+d_2\mathcal {A}[P](\xi)+ b_2K_2- P(\xi)I(\xi),\\
\mathcal{P}_3[S,P,I](\xi)&=(\mu-\delta)I(\xi)+d_3\mathcal {A}[I](\xi)+\gamma_1\displaystyle{\int_{0}^hf(\tau)S(\xi-c\tau)I(\xi-c\tau)\dif \tau}+\gamma_2P(\xi)I(\xi).
\end{aligned}
\end{equation*}

\begin{lemma}\label{L3}
For each $k>-\xi_3$,  the truncated problem \eqref{2.5} with boundary conditions \eqref{2.6} has
a unique solution $(S,P,I)\in [C(\mathbb{R})]^3\cap [C^1(\mathbb{R})/\{\pm k\}]^3$ such that
\begin{equation}\label{2.7}
	\sigma\leq S_-\leq S\leq K_1,\quad 0\leq P_-\leq P\leq K_2, \quad 0\leq I_-\leq I\leq I_+ \quad \text{ over } (-\infty,k].
\end{equation}
\end{lemma}
\begin{proof}
Define an operator $\mathcal {Q}_k=(\mathcal {Q}_{1k},\mathcal {Q}_{2k},\mathcal {Q}_{3k}): \Gamma_k\to \mathbb{X}_k$ as follows:
 \begin{equation*}
 \begin{aligned}
\mathcal {Q}_{1k}[S,P,I](\xi)
&=\me^{-\frac{\mu}{c}(\xi+k)} S_+(-k)+\frac{1}{c}\int_{-k}^{\xi} \me^{-\frac{\mu}{c}(\xi-z)}\mathcal{P}_1[S,P,I](z)\dif z, \\
 \mathcal {Q}_{2k}[S,P,I](\xi)
&=\me^{-\frac{\mu}{c}(\xi+k)} P_+(-k)+\frac{1}{c}\int_{-k}^{\xi} \me^{-\frac{\mu}{c}(\xi-z)}\mathcal{P}_2[S,P,I](z)\dif z,\\
\mathcal {Q}_{3k}[S,P,I](\xi)
&=\me^{-\frac{\mu}{c}(\xi+k)} I_+(-k)+\frac{1}{c}\int_{-k}^{\xi} \me^{-\frac{\mu}{c}(\xi-z)}\mathcal{P}_3[S,P,I](z)\dif z,  \quad \xi\in[-k,k].
\end{aligned}
\end{equation*}
 By Lemma \ref{L2},   then
\begin{equation*}
\begin{aligned}
\mathcal {Q}_{1k}[S_+,P,I_-](\xi)&\leq \me^{-\frac{\mu}{c}(\xi+k)} S_+(-k)+\frac{1}{c}\int_{-k}^{\xi} \me^{-\frac{\mu}{c}(\xi-z)}[\mu S_+(z)+cS_+'(z)] \dif z\\
&=\me^{-\frac{\mu}{c}(\xi+k)} S_+(-k)+S_+(\xi)-\me^{-\frac{\mu}{c}(\xi+k)}S_+(-k)=S_+(\xi)
\end{aligned}
\end{equation*}
and
\begin{equation*}
\begin{aligned}
\mathcal {Q}_{1k}[S_-,P,I_+](\xi)&\geq \me^{-\frac{\mu}{c}(\xi+k)} S_+(-k)+\frac{1}{c}\left\{\int_{-k}^{\xi_1}+\int_{\xi_1}^{\xi}\right\} \me^{-\frac{\mu}{c}(\xi-z)}[\mu S_-(z)+cS_-'(z)] \dif z\\
&=\me^{-\frac{\mu}{c}(\xi+k)} S_+(-k)+S_-(\xi)-\me^{-\frac{\mu}{c}(\xi+k)}S_-(-k)\geq S_-(\xi).
\end{aligned}
\end{equation*}
 Note that $\mathcal {Q}_{1k}$ is nondecreasing in its first variable while decreasing in its third variable. Thus
 \begin{equation*}
S_-(\xi)\leq Q_{1k}[S,P,I](\xi)\leq  S_+(\xi), \quad \xi\in[-k,k].
\end{equation*}
Similarly, we can prove $P_-(\xi)\leq Q_{2k}[S,P,I](\xi)\leq P_+(\xi)$ and $I_-(\xi)\leq Q_{3k}[S,P,I](\xi)\leq I_+(\xi)$ for all $\xi\in[-k,k]$.
On the other hand, $\mathcal {Q}_k[S,P,I](-k)=(S_+(-k), P_+(-k),I_+(-k))$. Therefore, $\mathcal {Q}_k(\Gamma_k)\subset \Gamma_k $.

Notice that  any fixed-point  of $\mathcal {Q}_k$ is a solution of  \eqref{2.5} and  \eqref{2.6}. We first show that the operator $\mathcal {Q}_k$ is completely continuous.  Since for  $\xi\in[-k,k]$,
\begin{equation*}\label{a1}
\begin{aligned}
|S'(\xi)|&=\left|-\frac{\mu}{c}S(\xi)+\frac{1}{c}\mathcal{P}_1[S,P,I](\xi)\right|\leq 2\frac{\mu}{c}K_1,\\
|P'(\xi)|&=\left|-\frac{\mu}{c}P(\xi)+\frac{1}{c}\mathcal{P}_2[S,P,I](\xi)\right|\leq 2\frac{\mu}{c}K_2,\\
|I'(\xi)|&=\left|-\frac{\mu}{c}I(\xi)+\frac{1}{c}\mathcal{P}_3[S,P,I](\xi)\right|\leq \left(2\frac{\mu}{c}+\lambda_1\right)\me^{\lambda_1k}
,
\end{aligned}
\end{equation*}
applying the Arzela-Ascoli's theorem on $[-k,k]$, we know that $\mathcal {Q}_k$ is compact. Next we prove  the operator $\mathcal {Q}_k$ is also continuous. Straightforward calculations give that
\begin{equation}\label{2.8}
	\begin{aligned}
	&\big|\mathcal{P}_1[S_1,P_1,I_1](\cdot)-\mathcal{P}_1[S_2,P_2,I_2](\cdot)\big|\\
	\leq  &\big|(\mu-d_1-b_1-I_1)(S_1-S_2)\big|+\big|(S_2-\sigma)(I_1-I_2)\big|\\
&+d_1\left|\sum_{|i|=1}^{N_1}J(i)\big[S_1(\xi-i)-S_2(\xi-i)\big]\right|\\
\leq &(\mu-b_1-I_1)\max_{[-k,k]}|S_1-S_2|+(S_2-\sigma)\max_{[-k,k]}|I_1-I_2|\\
\leq &\mu|S_1-S_2|_1+(K_1-\sigma)|I_1-I_2|_1\\
\leq &(\mu+K_1-\sigma)\|(S_1-S_2,P_1-P_2,I_1-I_2)\|_{\mathbb{X}_k}.
\end{aligned}
\end{equation}
From \eqref{2.8}, we have
\begin{equation*}
	\begin{aligned}
	&\big|\mathcal {Q}_{1k}[S_1,P_1,I_1](\xi)-\mathcal {Q}_{1k}[S_2,P_2,I_2](\xi)\big|\\
	\leq &\frac{1}{c}\int_{-k}^{\xi} \me^{-\frac{\mu}{c}(\xi-z)}\big|\mathcal{P}_1[S_1,P_1,I_1](z)-\mathcal{P}_1[S_2,P_2,I_2](z)\big|\dif z\\
\leq &\frac{\mu+K_1-\sigma}{\mu}\|(S_1-S_2,P_1-P_2,I_1-I_2)\|_{\mathbb{X}_k}.
	\end{aligned}
	\end{equation*}
Thus the operator $\mathcal {Q}_{1k}$ is continuous. Similarly, we can show $\mathcal {Q}_{2k}$ and $\mathcal {Q}_{3k}$ are also
continuous. At last, using the  Schauder's fixed-point theorem, the conclusion holds.
\end{proof}

The next result states the existence of semi-waves.
\begin{theorem}\label{T4}
Assume that (J), (f) and (H) hold. For $c>c^*$, the system \eqref{2.1} has a solution $(S,P,I)$ such that $\sigma<S<K_1$, $0<P<K_2$ and $I>0$ on
$\mathbb{R}$. Moreover
\begin{equation}\label{2.9}
	 \lim_{\xi\rightarrow-\infty}(S,P,I)(\xi)=(K_1,K_2,0).
\end{equation}
\end{theorem}

\begin{proof}
By means of  Lemma \ref{L3}, the Arzela-Ascoli's theorem and a successive extraction of
subsequences argument, there are some functions $S,P, I\in C^1(\mathbb{R})$ such  that $(S,P,I)$ is a nonnegative solution of \eqref{2.1}. Also, $(S,P,I)$ satisfies the inequality \eqref{2.7}. Hence, \eqref{2.9} holds.

Note $I\geq I_->0$ on $(-\infty,\xi_3)$. Assume some $\xi_4\in[\xi_3,\infty)$ is the rightmost point such that $I(\xi_4)=0$. Then $I'(\xi_4)=0$. By \eqref{2.1},
\begin{equation*}
\begin{aligned}
0=cI'(\xi_4)&=d_3\sum_{i\neq0}J(i)I(\xi_4-i)+\gamma_1 \displaystyle{\int_{0}^hf(\tau)S(\xi_4-c\tau)I(\xi_4-c\tau)\dif \tau}\\
&\geq d_3\sum_{i\neq0}J(i)I(\xi_4-i)\geq 0,
\end{aligned}
\end{equation*}
which yields that $I(\xi_4+ i)=0$ for all $i\in\mathbb{Z}_+$ with $J(i)>0$.  That's a contradiction to the definition of $\xi_4$.
Thus, $I>0$ on $\mathbb{R}$. Similarly, we can show that
 $\sigma<S<K_1$  and $0<P<K_2$ on $\mathbb{R}$.
\end{proof}

Note that the diffusion is nonlocal. In order to obtain the uniform upper and lower bounds of $\frac{I'(\xi)}{I(\xi)}$ on $\mathbb{R}$,
we require the kernel function $J(\cdot)$ is compactly supported. Then we have
\begin{lemma}\label{L5}
 $\frac{I'(\xi)}{I(\xi)}$ is uniformly bounded on $\mathbb{R}$.
\end{lemma}
\begin{proof}
Remember that
\begin{equation*}
\begin{aligned}
	&cI'(\xi)+(d_3+\delta)I(\xi)\\
= &d_3\sum_{i=1}^{N_1}J(i)[I(\xi+i)+I(\xi-i)]
+\gamma_1 \displaystyle{\int_{0}^hf(\tau)S(\xi-c\tau)I(\xi-c\tau)\dif \tau}+\gamma_2P(\xi)I(\xi).
\end{aligned}
\end{equation*}
Denote $u(\xi)=I(\xi)\me^{\frac{d_3+\delta }{c}\xi}$. Then $u(\xi)$ is  increasing on $\mathbb{R}$. Thus, $I(\xi-s)\me^{\frac{d_3+\delta }{c}(\xi-s)}\leq I(\xi)\me^{\frac{d_3+\delta }{c}\xi}$ for any $s>0$. Namely, $I(\xi-s)\leq I(\xi)\me^{\frac{d_3+\delta }{c}s}$ for any $s>0$. Especially,
for $s=i\in\mathbb{Z}_+$, we get $I(\xi-i)\leq I(\xi)\me^{\frac{d_3+\delta }{c}i}$. It implies that for $J(i)>0$,
\begin{equation}\label{2.12}
\sum_{i=1}^{N_1}J(i)I(\xi-i)\leq I(\xi)\sum_{i=1}^{N_1}J(i)\me^{\frac{d_3+\delta }{c}i},	  \quad \forall\;\xi\in\mathbb{R}.
\end{equation}
The inequality \eqref{2.12} shows that
\begin{equation}\label{2.13}
\frac{\sum_{i=1}^{N_1}J(i)I(\xi-i)}{I(\xi)}\leq \sum_{i=1}^{N_1}J(i)\me^{\frac{d_3+\delta }{c}i},	  \quad \forall\;\xi\in\mathbb{R}.
\end{equation}
Since
\begin{equation}\label{2.14}
\left[I(z)\me^{\frac{d_3+\delta }{c}z}\right]'\\
> \frac{d_3}{c}\me^{\frac{d_3+\delta }{c}z}\sum_{i=1}^{N_1}J(i)I(z+i),
\end{equation}
integrating the inequality \eqref{2.14} over $[\xi,\xi+\frac{i}{2}]$ with $1\leq i\leq N_1$, we reach
\begin{equation*}
\begin{aligned}
	&\me^{\frac{d_3+\delta }{c}(\xi+\frac{i}{2})}I(\xi+\frac{i}{2})-\me^{\frac{d_3+\delta }{c}\xi}I(\xi)\\
\geq &\frac{d_3}{c}\int_{\xi}^{\xi+\frac{i}{2}}\me^{\frac{d_3+\delta }{c}z}\sum_{i=1}^{N_1}J(i)I(z+i)\dif z\\
=&\frac{d_3}{c}\int_{\xi}^{\xi+\frac{i}{2}}\sum_{i=1}^{N_1}J(i)\me^{\frac{d_3+\delta }{c}(z+i)}I(z+i)\me^{-\frac{d_3+\delta }{c}i}\dif z\\
\geq &\frac{d_3}{c}\int_{\xi}^{\xi+\frac{i}{2}}\me^{\frac{d_3+\delta }{c}\xi}\sum_{i=1}^{N_1}J(i)I(\xi+i)\dif z\\
=&\frac{d_3i}{2c}\me^{\frac{d_3+\delta }{c}\xi}\sum_{i=1}^{N_1}J(i)I(\xi+i)\geq \frac{d_3}{2c}\me^{\frac{d_3+\delta }{c}\xi}\sum_{i=1}^{N_1}J(i)I(\xi+i).
\end{aligned}
\end{equation*}
Therefore, for  $1\leq i\leq N_1$,
\begin{equation*}
\begin{aligned}
	I(\xi+\frac{i}{2})&\geq \frac{d_3}{2c}\me^{-\frac{d_3+\delta }{c}\frac{i}{2}}\sum_{i=1}^{N_1}J(i)I(\xi+i)\\
&\geq \frac{d_3}{2c}\me^{-\frac{d_3+\delta }{c}\frac{N_1}{2}}\sum_{i=1}^{N_1}J(i)I(\xi+i).
\end{aligned}
\end{equation*}
It yields that
\begin{equation}\label{2.15}
\sum_{i=1}^{N_1}J(i)I(\xi+\frac{i}{2})\geq \frac{d_3}{4c}\me^{-\frac{d_3+\delta }{c}\frac{N_1}{2}}\sum_{i=1}^{N_1}J(i)I(\xi+i)
\end{equation}
due to $\sum_{i=1}^{N_1}J(i)=\frac{1}{2}$. The inequality \eqref{2.15} tells us that
\begin{equation}\label{2.16}
\frac{\sum_{i=1}^{N_1}J(i)I(\xi+i)}{\sum_{i=1}^{N_1}J(i)I(\xi+\frac{i}{2})}\leq \frac{4c}{d_3}\me^{\frac{d_3+\delta }{c}\frac{N_1}{2}},	  \quad \forall\;\xi\in\mathbb{R}.
\end{equation}
Similarly, by integrating the inequality \eqref{2.14} over $[\xi,\xi+i]$ with $1\leq i\leq N_1$, we also have
\begin{equation*}
\me^{\frac{d_3+\delta }{c}i}I(\xi+i)-I(\xi)
\geq \frac{d_3}{c} \sum_{i=1}^{N_1}J(i)I(\xi+i).
\end{equation*}
Hence,
\begin{equation}\label{2.17}
I(\xi+i)\geq \left[I(\xi)+\frac{d_3}{c}\sum_{i=1}^{N_1}J(i)I(\xi+i)\right]\me^{-\frac{d_3+\delta }{c}N_1}.
\end{equation}
Substituting \eqref{2.17} into \eqref{2.14}, we obtain
\begin{equation}\label{2.18}
\begin{aligned}
&\left[I(z)\me^{\frac{d_3+\delta }{c}z}\right]'\\
&\geq \frac{d_3}{c}\me^{\frac{d_3+\delta }{c}(z-N_1)}\sum_{i=1}^{N_1}J(i)\left[I(z)+\frac{d_3}{c}\sum_{i=1}^{N_1}J(i)I(z+i)\right]\\
&\geq \frac{d_3^2}{c^2}\me^{\frac{d_3+\delta }{c}(z-N_1)}\sum_{i=1}^{N_1}J(i)\left[\sum_{i=1}^{N_1}J(i)I(z+i)\right]\\
&=\frac{d_3^2}{2c^2}\me^{\frac{d_3+\delta }{c}(z-N_1)}\left[\sum_{i=1}^{N_1}J(i)I(z+i)\right].\\
\end{aligned}\end{equation}
Integrating the inequality \eqref{2.18} over $[\xi-\frac{i}{2},\xi]$ with $1\leq i\leq N_1$, then similarly as above,
\begin{equation*}
\me^{\frac{d_3+\delta }{c}\xi}I(\xi)
\geq \frac{d_3^2}{2c^2}\me^{-\frac{d_3+\delta }{c}N_1}\int_{\xi-\frac{i}{2}}^{\xi}\sum_{i=1}^{N_1}J(i)I(\xi+\frac{i}{2})\me^{\frac{d_3+\delta }{c}(\xi-\frac{i}{2})}\dif z.
\end{equation*}
Thus, for $1\leq i\leq N_1$,
\begin{equation*}
\begin{aligned}
	I(\xi)\geq &\frac{d_3^2}{2c^2}\me^{-\frac{d_3+\delta }{c}N_1}\cdot \frac{i}{2}\left[\sum_{i=1}^{N_1}J(i)I(\xi+\frac{i}{2})\me^{-\frac{d_3+\delta }{c}\frac{i}{2}}\right]\\
\geq& \frac{d_3^2}{4c^2}\me^{-\frac{d_3+\delta }{c}N_1}\sum_{i=1}^{N_1}J(i)I(\xi+\frac{i}{2})\me^{-\frac{d_3+\delta }{c}\frac{N_1}{2}}\\
=& \frac{d_3^2}{4c^2}\me^{-\frac{d_3+\delta }{c}\frac{3N_1}{2}}\sum_{i=1}^{N_1}J(i)I(\xi+\frac{i}{2}).
\end{aligned}
\end{equation*}
It implies that
\begin{equation}\label{2.19}
\frac{\sum_{i=1}^{N_1}J(i)I(\xi+\frac{i}{2})}{I(\xi)}\leq \frac{4c^2}{d_3^2}\me^{\frac{d_3+\delta }{c}\frac{3N_1}{2}},	  \quad \forall\;\xi\in\mathbb{R}.
\end{equation}
Inequalities \eqref{2.16} and \eqref{2.19} show that
\begin{equation}\label{2.20}
\begin{aligned}
\frac{\sum_{i=1}^{N_1}J(i)I(\xi+i)}{I(\xi)}&=\frac{\sum_{i=1}^{N_1}J(i)I(\xi+i)}{\sum_{i=1}^{N_1}J(i)I(\xi+\frac{i}{2})}
\frac{\sum_{i=1}^{N_1}J(i)I(\xi+\frac{i}{2})}{I(\xi)}\\
&\leq \frac{16c^3}{d_3^3}\me^{\frac{2(d_3+\delta)N_1 }{c}}.
\end{aligned}\end{equation}
Denote $\eta(\xi)=\frac{I'(\xi)}{I(\xi)}$. Then
\begin{equation*}
	\begin{aligned}
\eta(\xi)&\geq -\frac{d_3+\delta }{c}+\frac{\gamma_1 }{c} \displaystyle{\int_{0}^hf(\tau)S(\xi-c\tau)\frac{I(\xi-c\tau)}{I(\xi)}\dif \tau}\\
&\geq -\frac{d_3+\delta}{c} +\frac{\gamma_1 \sigma}{c} \displaystyle{\int_{0}^hf(\tau)\me^{\int_{\xi}^{\xi-c\tau}\eta(z)\dif z}\dif \tau}.
\end{aligned}
\end{equation*}
Put $w(\xi)=\me^{\frac{d_3+\delta }{c}\xi+\int_{0}^{\xi}\eta(z)\dif z}$.
Then $w(\xi)$ satisfies
\begin{equation*}
	\begin{aligned}
w'(\xi)&=w(\xi)\left[ \frac{d_3+\delta }{c}+\eta(\xi)\right]\geq\frac{\gamma_1 \sigma}{c} \displaystyle{\int_{0}^hf(\tau)\me^{\int_{\xi}^{\xi-c\tau}\eta(z)\dif z}\dif \tau}.
\end{aligned}
\end{equation*}
It implies $w(\xi)$ is increasing. Therefore,
\begin{equation*}
\displaystyle{\int_{0}^hf(\tau)\frac{I(\xi-c\tau)}{I(\xi)}\dif \tau}=\displaystyle{\int_{0}^hf(\tau)\me^{(d_3+\delta) \tau}\frac{w(\xi-c\tau)}{w(\xi)}\dif \tau}\leq \displaystyle{\int_{0}^hf(\tau)\me^{(d_3+\delta) \tau}\dif \tau},
\end{equation*}
which results in
\begin{equation}\label{2.21}
	0< \displaystyle{\int_{0}^hf(\tau)S(\xi-c\tau)\frac{I(\xi-c\tau)}{I(\xi)}\dif \tau}\leq K_1\displaystyle{\int_{0}^hf(\tau)\me^{(d_3+\delta) \tau}\dif \tau}.
\end{equation}
By inequalities \eqref{2.13}, \eqref{2.20}  and  \eqref{2.21}, then
\begin{equation*}
\begin{aligned}
	-\frac{(d_3+\delta)}{c}&<\frac{I'(\xi)}{I(\xi)}\\
&\leq \frac{d_3}{c}\left[\sum_{i=1}^{N_1}J(i)\me^{\frac{d_3+\delta }{c}i}+\frac{16c^3}{d_3^3}\me^{\frac{2(d_3+\delta)N_1 }{c}}\right]
+\frac{\gamma_1 K_1}{c} \displaystyle{\int_{0}^hf(\tau)\me^{(d_3+\delta) \tau}\dif \tau}+\frac{\gamma_2K_2}{c}.
\end{aligned}
\end{equation*}

\end{proof}

With the help of Lemma \ref{L5}, we are able to  verify that $I$ is uniformly bounded on $\mathbb{R}$.
\begin{lemma}\label{L6}
For any  sequence $\{c_\iota\}$ with $c_\iota\in (c^*,\tilde{c})$ , let $\{(S_\iota,P_\iota,I_\iota)\}$ be the solution sequence of \eqref{2.1} corresponding to wave speeds  $\{c_\iota\}$.
Then there exists some constant $M>0$ (independent of $\iota$) such that $\|I_\iota(\cdot)\|_{L^{\infty}(\mathbb{R})}\leq M$.
\end{lemma}

\begin{proof}
We will prove the lemma by three claims.

\textbf{Claim 1.} If there exists a sequence $\{\xi_\iota\}$ such that
$\lim_{\iota\to\infty}I_\iota(\xi_\iota)=\infty$, then $S_\iota(\xi+\xi_\iota)\to \sigma$ and  $P_\iota(\xi+\xi_\iota)\to 0$ in $C_{\text{loc}}(\mathbb{R})$ as $\iota\to\infty$.

 Note that
$$I_{\iota}(\xi+\xi_\iota)=I_{\iota}(\xi_\iota)\me^{\int_{\xi_\iota}^{\xi+\xi_\iota}\frac{I_{\iota}'(\theta)}{I_{\iota}(\theta)}\dif\theta}.$$
For any fixed $\nu\in\mathbb{R}$ with $\nu>0$,  since $\frac{I'_\iota(\xi)}{I_\iota(\xi)}$ is uniformly bounded on $\mathbb{R}$,
$\me^{\int_{\xi_\iota}^{\xi+\xi_\iota}\frac{I_{\iota}'(\theta)}{I_{\iota}(\theta)}\dif\theta}$ is uniformly bounded for any $\xi\in[-\nu,\nu]$. Thus,
if $I_\iota(\xi_\iota)\to\infty$ as $\iota\to\infty$, then $I_\iota(\xi+\xi_\iota)\to \infty$ for any $\xi\in[-\nu,\nu]$  as $\iota\to\infty$.
Notice we always get $S_\iota(\xi+\xi_\iota)\geq \sigma$ for any $\xi\in[-\nu,\nu]$ as $\iota\to\infty$.
Suppose $S_\iota(\xi+\xi_\iota)\not\to\sigma$ for $\xi\in[-\nu,\nu]$  as $\iota\to\infty$. Then up to extraction a subsequence, there exists a sequence $\{\vartheta_\iota\}$ with $\vartheta_\iota\in[-\nu,\nu]$, such
that $S_\iota(\vartheta_\iota+\xi_\iota)\geq \sigma+r'$ for some  positive constant $r'$.
Letting $\tilde{\xi}_\iota=\vartheta_\iota+\xi_\iota$, then $I_\iota(\tilde{\xi}_\iota)\to \infty$ as $\iota\to\infty$ and $S_\iota(\tilde{\xi}_\iota)\geq \sigma+r'$. Since
\begin{equation*}
\begin{aligned}
S_\iota'(\xi)&=\frac{d_1}{c}\mathcal {A}[S_\iota](\xi)+\frac{b_1}{c}[K_1-S_\iota(\xi)]- \frac{1}{c}[S_\iota(\xi)-\sigma]I_\iota(\xi)\\
&\leq \frac{d_1}{c}\sum_{i=1}^{N_1}J(i)\cdot 2K_1+\frac{b_1}{c}K_1- \frac{1}{c}[S_\iota(\xi)-\sigma]I_\iota(\xi)\\
&=\frac{(d_1+b_1)K_1}{c}- \frac{1}{c}[S_\iota(\xi)-\sigma]I_\iota(\xi),
\end{aligned}
\end{equation*}
it follows that for all $\xi\in[\tilde{\xi}_\iota-\varepsilon_0,\tilde{\xi}_\iota]$ with $\varepsilon_0=\frac{r'c}{2(d_1+b_1)K_1}$,
\begin{equation*}
S_\iota(\xi)=S_\iota(\tilde{\xi}_\iota)-\int_{\xi}^{\tilde{\xi}_\iota}S_\iota'(\theta)\dif \theta\geq \sigma+r'-\frac{(d_1+b_1)K_1}{c}\varepsilon_0=\sigma+\frac{r'}{2}.
\end{equation*}
Therefore for all $\xi\in[\tilde{\xi}_\iota-\varepsilon_0,\tilde{\xi}_\iota]$,
\begin{equation*}
S_\iota'(\xi)\leq \frac{(d_1+b_1)K_1}{c}- \frac{r'}{2c}I_\iota(\xi)\to-\infty, \quad \text{as } \iota\to\infty.
\end{equation*}
This contradicts the fact that $S_\iota$ is bounded. Thus due to the arbitrariness of $\nu$, $S_\iota(\xi+\xi_\iota)\to \sigma$ in $C_{\text{loc}}(\mathbb{R})$ as $\iota\to\infty$.
Similarly, we can show that $P_\iota(\xi+\xi_\iota)\to 0$ in $C_{\text{loc}}(\mathbb{R})$ as $\iota\to\infty$.

Letting $\{c_\iota\}=\{c\}$ and $(S_\iota,P_\iota,I_\iota)=(S,P,I)$, then we can prove the following claim.

\textbf{Claim 2.} If $\limsup_{\xi\to\infty} I(\xi)=\infty$,   then   $\lim_{\xi\to\infty} I(\xi)=\infty$ and  $I(\xi)$ is  strictly increasing for $\xi\gg 1$.

For otherwise, assume $\liminf_{\xi\to\infty} I(\xi)=:\underline{I}<\infty$. Thus, there exists a sequence $\{\vartheta_j\}$
satisfying $\vartheta_j\to\infty$ as $j\to\infty$ and $\lim_{j\to\infty} I(\vartheta_j)=\underline{I}$.
Note that $I(\vartheta_j)\leq \underline{I}+1$ for large $j$. For each $j\in \mathbb{Z}_+$, put some point $\xi_j\in (\vartheta_j,\vartheta_{j+1})$
such that $I(\xi_j)=\max_{\xi\in[\vartheta_j,\vartheta_{j+1}]}I(\xi)$. Since $\limsup_{\xi\to\infty} I(\xi)=\infty$,
$\lim_{j\to\infty}I(\xi_j)=\infty$. We next show that $[\xi_j-N_1-ch,\xi_j+N_1+ch]\subset [\vartheta_j,\vartheta_{j+1}]$.
According to Lemma \ref{L5}, set $\kappa=\sup_{\xi\in \mathbb{R}}|\eta(\xi)|$, where $\eta(\xi)=\frac{I'(\xi)}{I(\xi)}$.
Since $\lim_{j\to\infty}I(\xi_j)=\infty$, we suppose $I(\xi_j)\geq (\underline{I}+2) \me^{\kappa(N_1+ch)}$ for large $j$.
It follows that $\frac{I(\xi_j)}{I(\xi)}=\me^{\int_{\xi}^{\xi_j}\eta(s)\dif s}\leq \me^{\kappa|\xi_j-\xi|}\leq \me^{\kappa(N_1+ch)}$
if $\xi\in[\xi_j-N_1-ch,\xi_j+N_1+ch]$. Hence, $I(\xi)\geq I(\xi_j)\me^{-\kappa(N_1+ch)}\geq (\underline{I}+2)$ for large $j$. It implies
that $[\xi_j-N_1-ch,\xi_j+N_1+ch]\subset [\vartheta_j,\vartheta_{j+1}]$.
Thus $I'(\xi_j)=0$,  $d_3\mathcal {A}[I](\xi_j)\leq 0$ and $I(\xi_j-c\tau)\leq I(\xi_j)$ for $\tau\in[0,h]$. Setting $\xi=\xi_j$ in \eqref{2.1}, then
\begin{equation*}\label{aa}
\begin{aligned}
	0=cI'(\xi_j)&=d_3\mathcal {A}[I](\xi_j)+\gamma_1 \displaystyle{\int_{0}^hf(\tau)S(\xi_j-c\tau)I(\xi_j-c\tau)\dif \tau}- \delta I(\xi_j)+\gamma_2P(\xi_j)I(\xi_j)\\
&\leq \gamma_1 \displaystyle{\int_{0}^hf(\tau)S(\xi_j-c\tau)I(\xi_j)\dif \tau}- \delta I(\xi_j)+\gamma_2P(\xi_j)I(\xi_j)\\
&=  \displaystyle{\int_{0}^hf(\tau)[\gamma_1 S(\xi_j-c\tau)-\delta+\gamma_2P(\xi_j)]I(\xi_j)\dif \tau}.
\end{aligned}
\end{equation*}
By Claim 1, $\lim_{j\to\infty}S(\xi_j-c\tau)\to \sigma, \lim_{j\to\infty}P(\xi_j)\to 0$ if $\lim_{j\to\infty}I(\xi_j)\to \infty$. Remember $\gamma_1\sigma-\delta<0$. Hence, for large $j$, we have $[\gamma_1 S(\xi_j-c\tau)-\delta+\gamma_2P(\xi_j)]I(\xi_j)<0$. This contradicts the above inequality.
Since $I(\xi)>0$ on $\mathbb{R}$ and  $\lim_{\xi\to\infty} I(\xi)=\infty$, we see that $I'(\xi)>0$  for $\xi\gg 1$.

\textbf{Claim 3.} For any  $\iota\in\mathbb{Z}_+$, then $\|I_\iota(\cdot)\|_{L^{\infty}(\mathbb{R})}<\infty$.

 If not, there is some $\iota_0\in\mathbb{Z}_+$ such that $\|I_\iota(\cdot)\|_{L^{\infty}(\mathbb{R})}=\sup_{\xi\in\mathbb{R}}I_{\iota_0}(\xi)=\infty$.
Then by Claim 2, for any sequence $\{z_\iota\}$ with $\lim_{\iota \to\infty}z_\iota=\infty$, we get
 $I_{\iota_0}(z_{\iota})\to \infty$ as $\iota\to\infty$. Putting $\bar{I}_{\iota}(\xi)=\frac{I_{\iota_0}(\xi+z_\iota)}{I_{\iota_0}(z_\iota)}$,
then $\bar{I}_{\iota}(\xi)$ is uniformly bounded for $\xi$ in any compact set. Hence, $\bar{I}_{\iota}'(\xi)=\frac{I_{\iota_0}'(\xi+z_\iota)}{I_{\iota_0}(z_\iota)}
=\frac{I_{\iota_0}'(\xi+z_\iota)}{I_{\iota_0}(\xi+z_\iota)}\bar{I}_{\iota}(\xi)$ is also uniformly bounded for $\xi$ in any compact set.
It implies that $\bar{I}_{\iota}(\xi)$ is equi-continuous for $\xi$ in any compact set.
Thus, $\bar{I}_{\iota}(\xi)\to \bar{I}_{\infty}(\xi)$ in $C_{\text{loc}}(\mathbb{R})$ as $\iota\to\infty$.
Note that $\bar{I}_{\iota}(\xi)$ satisfies
\begin{equation}\label{2.22}
\begin{aligned}
	&c_{\iota_0}\bar{I}_{\iota}'(\xi)+(d_3+\delta)\bar{I}_{\iota}(\xi)-\gamma_2P_{\iota_0}(z_{\iota}+\xi)\bar{I}_{\iota}(\xi)\\
= &d_3\sum_{i=1}^{N_1}J(i)[\bar{I}_{\iota}(\xi+i)+\bar{I}_{\iota}(\xi-i)]
+\gamma_1 \displaystyle{\int_{0}^hf(\tau)S_{\iota_0}(z_{\iota}+\xi-c_{\iota_0}\tau)\bar{I}_{\iota}(\xi-c_{\iota_0}\tau)\dif \tau}.
\end{aligned}
\end{equation}
Claim 1 shows that $S_{\iota_0}(z_{\iota}+\xi-c_{\iota_0}\tau)\to \sigma, P_{\iota_0}(z_{\iota}+\xi)\to 0$ in $C_{\text{loc}}(\mathbb{R})$ as $\iota\to\infty$.
Thus, by \eqref{2.22}, $\{\bar{I}_{\iota}'(\cdot)\}$ is uniformly convergent  in $C_{\text{loc}}(\mathbb{R})$, and $\bar{I}_{\iota}(0)=1$ converges, then  $\bar{I}_{\infty}(\xi)$ is differentiable and  $\bar{I}_{\iota}'(\cdot)\to \bar{I}_{\infty}'(\cdot)$
as $\iota\to\infty$.
Letting $\iota \to\infty$ in \eqref{2.22}, we have
\begin{equation}\label{2.23}
\begin{aligned}
	&c_{\iota_0}\bar{I}_{\infty}'(\xi)+(d_3+\delta)\bar{I}_{\infty}(\xi)\\
= &d_3\sum_{i=1}^{N_1}J(i)[\bar{I}_{\infty}(\xi+i)+\bar{I}_{\infty}(\xi-i)]
+\gamma_1 \sigma\displaystyle{\int_{0}^hf(\tau)\bar{I}_{\infty}(\xi-c_{\iota_0}\tau)\dif \tau}.
\end{aligned}
\end{equation}
Denote
\begin{equation*}
\triangle_{\sigma}(\lambda,c_{\iota_0})=d_3\sum_{i=1}^{N_1}J(i)[\me^{\lambda i}+\me^{-\lambda i}]-c_{\iota_0}\lambda+\gamma_1 \sigma\int_{0}^h f(\tau)\me^{-\lambda c_{\iota_0}\tau}\dif \tau-(d_3+\delta).
\end{equation*}
Since $\triangle_{\sigma}(0,c_{\iota_0})<0,\partial_{\lambda}\triangle_{\sigma}(0,c_{\iota_0})<0$, $\triangle_{\sigma}(\pm\infty,c_{\iota_0})=\infty, \partial_{\lambda\lambda}\triangle_{\sigma}(\lambda,c_{\iota_0})>0$, we see that for any $c_{\iota_0}>0$, there is one  negative real root $\eta_1=\eta_1(c_{\iota_0})$ and one  positive real root $\eta_2=\eta_2(c_{\iota_0})$ of $\triangle_{\sigma}(\lambda,c_{\iota_0})=0$.
By the similar proofs like Theorem 3.1 in \cite{chen}, the solution  $\bar{I}_{\infty}(\xi)$ of
\eqref{2.23} can be given by  $\bar{I}_{\infty}(\xi)=\omega\me^{\eta_1\xi}+(1-\omega) \me^{\eta_2\xi}, \omega\in[0,1]$. If $\omega>0$, then $\bar{I}_{\infty}(\xi)\to \infty $ as $\xi\to-\infty$.
Claim 2 implies that $I_{\iota_0}(z)$ is strictly increasing for $z\gg 1$. Hence, for $\xi\leq 0$,
 $\bar{I}_{\infty}(\xi)=\lim_{\iota \to\infty}\frac{I_{\iota_0}(\xi+z_\iota)}{I_{\iota_0}(z_\iota)}\leq 1$. It is
impossible. Hence, $\omega\equiv0$ and $\bar{I}_{\infty}(\xi)=\me^{\eta_2\xi}$.
Then $\bar{I}_{\infty}'(0)=\eta_2=\lim_{\iota \to\infty}\frac{I_{\iota_0}'(z_\iota)}{I_{\iota_0}(z_\iota)}$.
Since $\{z_\iota\}$ is an arbitrary sequence, Heine's theorem  yields that $\lim_{\xi \to\infty}\frac{I_{\iota_0}'(\xi)}{I_{\iota_0}(\xi)}=\eta_2$.
Thus, $I_{\iota_0}(\xi)\sim \me^{\eta_2\xi}$ as $\xi\to\infty$.
On the other hand, for $c_{\iota_0}>c^*$, it follows from Lemma \ref{L1} (iii) that
\begin{equation*}
d_3\sum_{i=1}^{N_1}J(i)[\me^{\lambda_1 i}+\me^{-\lambda_1 i}]-c_{\iota_0}\lambda_1+\gamma_1 K_1\int_{0}^h f(\tau)\me^{-\lambda_1 c_{\iota_0}\tau}\dif \tau+\gamma_2K_2-(d_3+\delta)=0.
\end{equation*}
Therefore,
\begin{equation*}
\triangle_{\sigma}(\lambda_1(c_{\iota_0}),c_{\iota_0})=\gamma_1 (\sigma-K_1)\int_{0}^h f(\tau)\me^{-\lambda_1 c_{\iota_0}\tau}\dif \tau-\gamma_2K_2<0.
\end{equation*}
It follows that $0<\lambda_1<\eta_2$. Since   $I_{\iota_0}(\xi)\leq \me^{\lambda_1\xi}$ on  $\mathbb{R}$, we get a  contradiction.

 Finally, we  show that there exists a positive constant $M$ such that $\|I_\iota(\cdot)\|_{L^{\infty}(\mathbb{R})}\leq M$ for any  $\iota\in\mathbb{Z}_+$. Set $M_{\iota}=\sup_{\xi\in\mathbb{R}}I_{\iota}(\xi)$. If $\lim_{\iota\to\infty}M_{\iota}=\infty$,
 then there exists a sequence $\{\vartheta_\iota\}$
satisfies $\lim_{\iota\to\infty} I_{\iota}(\vartheta_\iota)=\infty$ and $I_{\iota}(\vartheta_\iota)\geq \varepsilon_1M_{\iota}$  for $\iota$
large enough and some $\varepsilon_1\in(0,1)$. Claim 2 shows that $\lim_{\iota\to\infty}\vartheta_\iota=\infty$.
Then $\check{I}_{\iota}(\xi)=\frac{I_{\iota}(\xi+\vartheta_\iota)}{I_{\iota}(\vartheta_\iota)}\leq \frac{M_{\iota}}{\varepsilon_1M_{\iota}}=\frac{1}{\varepsilon_1}$.
In a similar way as above, $\check{I}_{\iota}(\xi)\to\check{I}_{\infty}(\xi)$
as $\iota\to\infty$. Thus $\check{I}_{\infty}\leq \frac{1}{\varepsilon_1}$ on  $\mathbb{R}$. Since the sequence $\{c_\iota\}$
is bounded, there exists a convergent subsequence. Assume this  subsequence converges to $c_\infty$. Then $c_\infty\in [c^*,\tilde{c}]$.
Notice that  $\check{I}_{\infty}$ satisfies \eqref{2.23} when  $c_{\iota_0}$ is replaced by $c_\infty$.
Hence, $\check{I}_{\infty}(\xi)=\omega'\me^{\tilde{\eta}_1\xi}+(1-\omega') \me^{\tilde{\eta_2}\xi}, \omega'\in[0,1]$  is an unbounded solution of \eqref{2.23} when $c_{\iota_0}$ is replaced by  $c_\infty$. This contradicts   $\check{I}_{\infty}\leq \frac{1}{\varepsilon_1}$.

\end{proof}

In order to prove $\liminf_{\xi\rightarrow\infty}I(\xi)>0$, we will use the Laplace transform. To perform the  Laplace transform
of $I$, we first make some prior estimates.

\begin{lemma}\label{L7}
There exists two positive constants   $\rho$ and  $M_2$ such that  $I(\xi)\me^{-\rho\xi}\leq M_2$ for all $\xi\in\mathbb{R}$.
\end{lemma}

\begin{proof}
We first prove $I\in L^1(-\infty,\xi)$. Choose a constant $\varepsilon_2$
such that $0<\varepsilon_2<K_2$ and $q:=\gamma_1(K_1-\varepsilon_2)>\delta$. Notice that $\lim_{\xi\rightarrow-\infty}S(\xi)=K_1$ and $\lim_{\xi\rightarrow-\infty}P(\xi)=K_2$.
There is $\varsigma_*\gg 1$  such that $S(\xi)>K_1-\varepsilon_2$ and $P(\xi)>K_2-\varepsilon_2$ for $\xi<-\varsigma_*$.
Put $\xi_*=-\varsigma_*-N_1$. Integrating the I-equation of \eqref{2.1} from $-\infty$ to $\xi$ with $\xi\leq \xi_*$, then
\begin{equation}\label{2.24}
\begin{aligned}
	cI(\xi)&=d_3\sum_{|i|=1}^{N_1}J(i)\int_{\xi}^{\xi-i}I(z)\dif z-\delta\int_{-\infty}^{\xi}I(z)\dif z\\
&\qquad+\int_{-\infty}^{\xi}\int_{0}^hf(\tau)\gamma_1 S(z-c\tau)I(z-c\tau)\dif \tau\dif z+\int_{-\infty}^{\xi}\gamma_2P(z)I(z)\dif z\\
&\geq d_3\sum_{|i|=1}^{N_1}J(i)\int_{\xi}^{\xi-i}I(z)\dif z+[\gamma_2(K_2-\varepsilon_2)-\delta]\int_{-\infty}^{\xi}I(z)\dif z+\int_{-\infty}^{\xi}\int_{0}^hf(\tau)qI(z-c\tau)\dif \tau\dif z\\
&= d_3\sum_{|i|=1}^{N_1}J(i)\int_{\xi}^{\xi-i}I(z)\dif z+q_1\int_{-\infty}^{\xi}I(z)\dif z+q\int_{-\infty}^{\xi}\int_{0}^hf(\tau)[I(z-c\tau)-I(z)]\dif \tau\dif z,
\end{aligned}
\end{equation}
where $q_1:=q-\delta+\gamma_2(K_2-\varepsilon_2)$.
Note that $I(-\infty)=0$. Then
\begin{equation}\label{2.25}
\begin{aligned}
\int_{-\infty}^{\xi}[I(z-c\tau)-I(z)]\dif  z&=\lim_{l\to-\infty}\int_{l}^{\xi}[I(z-c\tau)-I(z)]\dif  z\\
&=\lim_{l\to-\infty}\int_{l}^{\xi}\int_{0}^1-c\tau I'(z-c\tau s)\dif s\dif  z\\
&=-c\tau\lim_{l\to-\infty}\int_{0}^1[I(\xi-c\tau s)-I(l-c\tau s)]\dif s\\
&=-c\tau\int_{0}^1I(\xi-c\tau s)\dif s.
\end{aligned}
\end{equation}
Substituting \eqref{2.25} into \eqref{2.24}, we get
\begin{equation}\label{2.26}
cI(\xi)\geq d_3\sum_{|i|=1}^{N_1}J(i)\int_{\xi}^{\xi-i}I(z)\dif z+q_1\int_{-\infty}^{\xi}I(z)\dif z-cq\int_{0}^h\tau f(\tau)\int_{0}^1I(\xi-c\tau s)\dif s\dif \tau.
\end{equation}
Thus, by Lemma \ref{L6},
\begin{equation*}
\begin{aligned}
	q_1\int_{-\infty}^{\xi}I(z)\dif z&\leq cI(\xi)-d_3\sum_{|i|=1}^{N_1}J(i)\int_{\xi}^{\xi-i}I(z)\dif z\\
&\qquad+cq\int_{0}^h\tau f(\tau)\int_{0}^1I(\xi-c\tau s)\dif s\dif \tau\\
&\leq cM+d_3M\sum_{|i|=1}^{N_1}J(i)|i|+cqM\int_{0}^h\tau f(\tau)\dif \tau.
\end{aligned}
\end{equation*}
This implies  $I\in L^1(-\infty,\xi)$. Define $V(\xi)=\int_{-\infty}^{\xi}I(z)\dif z$. Then $V$ is positive, continuous and
increasing on $\mathbb{R}$. Integrating \eqref{2.26} from $-\infty$ to $\xi$ with $\xi\leq \xi_*$, we have
\begin{equation}\label{2.27}
\begin{aligned}
	cV(\xi)&\geq d_3\sum_{|i|=1}^{N_1}J(i)\int_{\xi}^{\xi-i}V(z)\dif z+q_1\int_{-\infty}^{\xi}V(z)\dif z\\
&\qquad-cq\int_{0}^h\tau f(\tau)\int_{0}^1V(\xi-c\tau s)\dif s\dif \tau.
\end{aligned}
\end{equation}
Since $V$ is increasing, then
\begin{equation}\label{2.28}
\sum_{|i|=1}^{N_1}J(i)\int_{\xi}^{\xi-i}V(z)\dif z=\sum_{i=1}^{N_1}J(i)\left[\int_{\xi}^{\xi+i}V(z)\dif z-
\int_{\xi-i}^{\xi}V(z)\dif z\right]>0.
\end{equation}
It follows from \eqref{2.27} and \eqref{2.28} that
\begin{equation*}
cV(\xi)+cqV(\xi)\int_{0}^h\tau f(\tau)\dif \tau\geq q_1\int_{-\infty}^{\xi}V(z)\dif z.
\end{equation*}
Setting $f_h=\int_{0}^h\tau f(\tau)\dif \tau$, then
\begin{equation*}
(c+cqf_h)V(\xi)\geq q_1\int_{-\theta}^{0}V(\xi+z)\dif z\geq q_1V(\xi-\theta)\cdot \theta
\end{equation*}
for any $\theta>0$. Hence, for $\theta_0>\frac{c+cqf_h}{q_1}=:q_2$, $V(\xi-\theta_0)\leq \frac{q_2}{\theta_0}V(\xi)=:\theta_1V(\xi)$,
and $\theta_1\in(0,1)$. Let $\rho=\frac{1}{\theta_0}\ln\frac{1}{\theta_1}$ and $V_{\rho}(\xi)=V(\xi)\me^{-\rho\xi}$.
Then for $\xi\leq \xi_*$,
\begin{equation*}
\begin{aligned}
	V_{\rho}(\xi-\theta_0)&= V(\xi-\theta_0)\me^{-\rho\xi}\me^{\rho\theta_0}=V(\xi-\theta_0)\me^{-\rho\xi}\me^{\ln\frac{1}{\theta_1}}\\
&=\frac{1}{\theta_1}V(\xi-\theta_0)\me^{-\rho\xi}\leq \frac{1}{\theta_1}\theta_1V(\xi)\me^{-\rho\xi}=V_{\rho}(\xi).
\end{aligned}
\end{equation*}
It yields that $V_{\rho}(\xi)$ is bounded  as $\xi\to -\infty$. Thus, there exists some $V_{\rho}^0$ such that
\begin{equation}\label{2.29}
V(\xi)\me^{-\rho\xi}\leq V_{\rho}^0, \quad\text{ for }\xi\leq \xi_*.
\end{equation}
On the other hand,
for $\xi>\xi_*$, we obtain
\begin{equation*}
V(\xi)=\int_{-\infty}^{\xi}I(z)\dif z\leq \int_{-\infty}^{\xi^*}I(z)\dif z+ \int^0_{\xi^*}I(z)\dif z+\int_{0}^{|\xi|}I(z)\dif z.
\end{equation*}
Since $I\in L^1(-\infty,\xi^*)$, there exists $M_1>0$ such that
\begin{equation}\label{2.30}
V(\xi)\me^{-\rho\xi}\leq \left[\int_{-\infty}^{\xi^*}I(z)\dif z+ \int^0_{\xi^*}I(z)\dif z+M|\xi|\right]\me^{-\rho\xi}\leq M_1, \quad\text{ for }\xi>\xi_*.
\end{equation}
Note that
\begin{equation*}
\begin{aligned}
	cI(\xi)\leq d_3\sum_{|i|=1}^{N_1}J(i)[V(\xi-i)-V(\xi)]
+\gamma_1 K_1\int_{0}^hf(\tau)V(\xi-c\tau)\dif \tau+(\gamma_2 K_2-\delta)V(\xi).
\end{aligned}
\end{equation*}
This, combined with \eqref{2.29}, \eqref{2.30} show that $I(\xi)\me^{-\rho\xi}\leq M_2$ for all $\xi\in \mathbb{R}$.
\end{proof}

The following lemma expresses that $I$ will be  increasing once  it turns to small.
\begin{lemma}\label{L8}
Let $0<c_1\leq c_2$  be given two constants and  $(S,P,I)$ be a  solution of \eqref{2.1} corresponding to $c\in[c_1,c_2]$. Further assume  $\sigma<S<K_1$, $0<P<K_2$ and $I>0$. There is some  $\varsigma>0$ such that $I'(\xi)>0$ as long as
 $I(\xi)\leq \varsigma$ for $\xi\in \mathbb{R}$.
\end{lemma}

\begin{proof}
Assume no such constant $\varsigma$  exists. Then there is a sequence  $\{c_\iota\}$  in $[c_1,c_2]$  and a solution sequence $\{(S_\iota,P_\iota,I_\iota)\}$  of \eqref{2.1} corresponding to wave speeds  $\{c_\iota\}$ such that
 \begin{equation}\label{2.31}
\lim_{\iota\to\infty}I_\iota(\xi_\iota)=0,   \quad I'_\iota(\xi_\iota)\leq 0 \text{ for all } \iota \in\mathbb{Z}_+
\end{equation}
for some $\{\xi_\iota\}$. Up to extraction of a subsequence, we may suppose  $\lim_{\iota\to\infty}c_\iota=c_\infty$.
As in the proof of Lemma \ref{L6},
we obtain  $(S_{\iota}(\xi+\xi_\iota),P_{\iota}(\xi+\xi_\iota),I_{\iota}(\xi+\xi_\iota))\to (S_{\infty}(\xi), P_{\infty}(\xi),I_{\infty}(\xi))$ in $[C^1_{\text{loc}}(\mathbb{R})]^3$ as $\iota\to\infty$. Note that $(S_{\infty}(\cdot),P_{\infty}(\cdot),I_{\infty}(\cdot))$ is also a solution of  \eqref{2.1} by replacing $c$ with $c_\infty$.
We next prove the following proposition.

\textbf{Proposition 1.} $S_{\infty}(\xi)\equiv K_1$, $P_{\infty}(\xi)\equiv K_2$ and $I_{\infty}(\xi)\equiv 0$ on $\mathbb{R}$.

It is easy to see that in any compact set of $\mathbb{R}$, $I_{\infty}(\xi)$ satisfies the following equation
\begin{equation}\label{2.31a}
\begin{aligned}
	&c_{\infty}I_{\infty}'(\xi)+(d_3+\delta)I_{\infty}(\xi)\\
= &d_3\sum_{|i|=1}^{N_1}J(i)I_{\infty}(\xi-i)
+\gamma_1\displaystyle{\int_{0}^hf(\tau)S_{\infty}(\xi-c_{\infty}\tau)I_{\infty}(\xi-c_{\infty}\tau)\dif \tau}+\gamma_2 P_{\infty}(\xi)I_{\infty}(\xi).
\end{aligned}
\end{equation}
Since $\lim_{\iota\to\infty}I_{\iota}(\xi_\iota)=0$, we get $I_{\infty}(0)=0$ and $I_{\infty}'(0)=0$. Then it follows from \eqref{2.31a} that
\begin{equation*}
0=d_3\sum_{|i|=1}^{N_1}J(i)I_{\infty}(-i)+\gamma_1 \displaystyle{\int_{0}^hf(\tau)S_{\infty}(-c_{\infty}\tau)I_{\infty}(-c_{\infty}\tau)\dif \tau}.
\end{equation*}
This yields that $I_{\infty}(i)=0$ for all $i\in\mathbb{Z}$ by induction.  Since $c_{\infty}I_{\infty}'(\xi)+(d_3+\delta)I_{\infty}(\xi)\geq 0$, the nonnegative function $u_{\infty}(\xi)=I_{\infty}(\xi)\me^{\frac{d_3+\delta}{c_{\infty}}\xi}$ is increasing on $\mathbb{R}$.
On the other hand, $u_{\infty}(i)=0$ for all $i\in\mathbb{Z}$. Thus,
$I_{\infty}(\xi)\equiv 0$ on $\mathbb{R}$.
Note that $S_{\infty}(\xi)$ satisfies
\begin{equation}\label{2.32}
c_{\infty}S'_\infty(\xi)=d_1\sum_{|i|=1}^{N_1}J(i)[S_\infty(\xi-i)-S_\infty(\xi)]+ b_1[K_1-S_\infty(\xi)],
\end{equation}
and $\sigma\leq S_\infty(\xi)\leq K_1$ for any $\xi\in\mathbb{R}$. Choosing $c_{\infty}\mu>d_1+b_1$, it follows from
\eqref{2.32} that
\begin{equation*}
S_\infty(\xi)=\me^{-\mu\xi}\int_{-\infty}^{\xi} \me^{\mu z}H[S_\infty](z)\dif z,
\end{equation*}
where $H[S_\infty](\xi)=\big[\mu-\frac{d_1}{c_{\infty}}-\frac{b_1}{c_{\infty}}\big]S_\infty(\xi)
+\frac{d_1}{c_{\infty}}\sum_{|i|=1}^{N_1}J(i)S_\infty(\xi-i)+ \frac{b_1}{c_{\infty}}K_1$. Define
\begin{equation*}
G[S_\infty](\xi)=\me^{-\mu\xi}\int_{-\infty}^{\xi} \me^{\mu z}H[S_\infty](z)\dif z.
\end{equation*}
The definition of $H$ implies that $H[S_\infty]$ is monotone increasing with respect to the function $S_\infty$ for $\sigma\leq S_\infty\leq K_1$.
We can further deduce that, if $S_\infty(\xi)$ is increasing, i.e., $S_\infty(\xi)\le S_\infty(\xi+\theta)$ for any  constant $\theta>0$ and all $\xi\in\mathbb{R}$, then
$$H[S_\infty(\cdot)] (\xi)\leq H[S_\infty(\theta+\cdot)](\xi)=H[S_\infty(\cdot)](\theta+\xi)$$
for all $\xi\in\mathbb{R}$ and any $\theta>0$. Thus $H[S_\infty](\xi)$ is increasing in $\xi$ whenever $S_\infty(\xi)$ is increasing in $\xi$.
Let $S_{\infty}^{0}(\xi)\equiv \sigma$ for all $\xi\in\mathbb{R}$. Define $\psi_0=S_{\infty}^{0}, \psi_m=G[\psi_{m-1}]$. Then it follows from
\begin{equation*}
G[S_{\infty}^{0}](\xi)=\me^{-\mu\xi}\int_{-\infty}^{\xi} \me^{\mu z}\left[\mu\sigma+\frac{b_1}{c_{\infty}}(K_1-\sigma)\right]\dif z>\sigma=S_{\infty}^{0}(\xi)
\end{equation*}
that $\psi_m$ is a monotone increasing sequence  in $C(\mathbb{R},\mathbb{R})$.
It also satisfies
$$S_{\infty}^{0}(\xi)\leq\cdots\leq \psi_m(\xi)\leq \psi_{m+1}(\xi)\leq S_{\infty}(\xi)\leq K_1,\quad \forall\;\xi\in\mathbb{R},$$
since $S_{\infty}(\xi)$ is a solution of the equation (\ref{2.32}) satisfying
$S_{\infty}(\xi)=G[S_{\infty}(\xi)]$ and $\psi_0\leq S_{\infty}(\xi)\leq K_1$ over $\mathbb{R}$. Hence there exists a function $\tilde{\varphi}\in C(\mathbb{R},\mathbb{R})$ such that $\tilde{\varphi}=\lim_{m\to\infty} \psi_m(\xi)$ and
\begin{equation}\label{2.33}
\sigma=S_{\infty}^{0}(\xi)\leq
\tilde{\varphi}(\xi)\leq S_{\infty}(\xi)\leq K_1, \quad \forall\;\xi\in\mathbb{R}.
\end{equation}
Moreover $\tilde{\varphi}(\xi)$ is monotone increasing in $\xi$ since $\psi_m(\xi)$ is a increasing function of $\xi$ for each $m$.
Hence both $\lim_{\xi\to \pm\infty}\tilde{\varphi}(\xi)
=\tilde{\varphi}_{\pm}$ exist and $\sigma\leq \tilde{\varphi}_{\pm}\leq K_1$. It is obvious that both $\tilde{\varphi}_+$ and $\tilde{\varphi}_-$ must be positive constant solutions of the equation (\ref{2.32}). It then follows that
$$\tilde{\varphi}_+=\tilde{\varphi}_-=K_1$$
since (\ref{2.32}) has a unique positive constant solution $K_1$. Thus the monotonicity of $\tilde{\varphi}(\xi)$ yields that
$\tilde{\varphi}(\xi)\equiv K_1$. Finally the inequality (\ref{2.33}) yields that  $S_{\infty}(\xi)\equiv K_1$. Similarly, we can show that  $P_{\infty}(\xi)\equiv K_2$.

Denoting
$\hat{I}_{\iota}(\xi)=\frac{I_{\iota}(\xi+\xi_\iota)}{I_{\iota}(\xi_\iota)}.$
Then in a similar way as in Lemma  \ref{L6} (Claim 3), $\hat{I}_{\iota}(\xi)\to\hat{I}_{\infty}(\xi)$ in $C^1_{\text{loc}}(\mathbb{R})$
as $\iota\to\infty$.
It is easy to see that in any compact set of $\mathbb{R}$, $\hat{I}_{\infty}(\xi)$ satisfies the following equation
\begin{equation*}
\begin{aligned}
	&c_{\infty}\hat{I}_{\infty}'(\xi)+(d_3+\delta-\gamma_2K_2)\hat{I}_{\infty}(\xi)\\
= &d_3\sum_{i=1}^{N_1}J(i)[\hat{I}_{\infty}(\xi+i)+\hat{I}_{\infty}(\xi-i)]
+\gamma_1 K_1\displaystyle{\int_{0}^hf(\tau)\hat{I}_{\infty}(\xi-c_{\infty}\tau)\dif \tau}.
\end{aligned}
\end{equation*}
Notice $\hat{I}_{\infty}(0)=1$. It is easy to check $\hat{I}_{\infty}(\xi)>0$ for any $\xi\in\mathbb{R}$. Denote $Z(\xi)=\frac{\hat{I}_{\infty}'(\xi)}{\hat{I}_{\infty}(\xi)}$.
Then $Z(\xi)$ satisfies
\begin{equation}\label{2.34}
\begin{aligned}
	c_{\infty}Z(\xi)&=d_3\sum_{i=1}^{N_1}J(i)\left[\me^{\int_{\xi}^{\xi+i}Z(s)\dif s}+\me^{\int_{\xi}^{\xi-i}Z(s)\dif s}\right]\\
&\quad+\gamma_1 K_1\displaystyle{\int_{0}^hf(\tau)\me^{\int_{\xi}^{\xi-c_{\infty}\tau}Z(s)\dif s}\dif \tau}+\gamma_2K_2- (d_3+\delta).
\end{aligned}
\end{equation}
Combined \cite[Lemma 3.4]{chen} with \cite[Lemma 3.4]{zlw}, $Z(\xi)$ has finite limits $z_{\pm}$ as $\xi\to\pm\infty$, which are roots of the characteristic equation
\begin{equation}\label{2.34a}
c_{\infty}z_{\pm}=d_3\sum_{i=1}^{N_1}J(i)[\me^{iz_{\pm}}+\me^{-iz_{\pm}}]+\gamma_1 K_1\int_{0}^h f(\tau)\me^{-z_{\pm} c_{\infty}\tau}\dif \tau+\gamma_2K_2-(d_3+\delta).
\end{equation}
If $z_{\pm}\leq 0$, then  $c_{\infty}z_{\pm}\leq 0$ due to $c_\infty\geq c_1>0$. But the right-hand side of \eqref{2.34a} is positive since
$$d_3\sum_{i=1}^{N_1}J(i)[\me^{iz_{\pm}}+\me^{-iz_{\pm}}]-d_3\geq 0$$ and
$$\gamma_1 K_1\int_{0}^h f(\tau)\me^{-z_{\pm} c_{\infty}\tau}\dif \tau+\gamma_2K_2-\delta\geq \gamma_1 K_1-\delta+\gamma_2K_2>\gamma_2K_2.$$
This is a contradiction. Thus $z_{\pm}>0$. By differentiating with respect to $\xi$ in \eqref{2.34}, we see
\begin{equation}\label{2.34b}
\begin{aligned}
	c_{\infty}Z'(\xi)&=d_3\sum_{i=1}^{N_1}J(i)\left[\me^{\int_{\xi}^{\xi+i}Z(s)\dif s}(Z(\xi+i)-Z(\xi))+\me^{\int_{\xi}^{\xi-i}Z(s)\dif s}(Z(\xi-i)-Z(\xi))\right]\\
&\quad+\gamma_1 K_1\displaystyle{\int_{0}^hf(\tau)\me^{\int_{\xi}^{\xi-c_{\infty}\tau}Z(s)\dif s}(Z(\xi-c_{\infty}\tau)-Z(\xi))\dif \tau}.
\end{aligned}
\end{equation}
If $Z(\xi)$ attains the minimum  at the point $\zeta^*$, then $Z'(\zeta^*)=0$. Set $\xi=\zeta^*$ in \eqref{2.34b}.
The left-hand side of \eqref{2.34b} is zero at $\zeta^*$ implies $Z(\zeta^*\pm i)=Z(\zeta^*-c_{\infty}\tau)=Z(\zeta^*)$. By induction, it follows that $Z$  is a   positive constant by \eqref{2.34}. Therefore, $\inf_{\mathbb{R}}Z(\cdot)\geq \min\{z_{\pm},Z(\zeta^*)\}>0.$
It yields that $\hat{I}_{\infty}'(\xi)>0$ for all $\xi\in\mathbb{R}$. In particularly,
 $0<\hat{I}_{\infty}'(0)=\lim_{\iota \to\infty}\frac{I_{\iota}'(\xi_\iota)}{I_{\iota}(\xi_\iota)}.$
It follows that for $\iota$ sufficiently large, $I_{\iota}'(\xi_\iota)>0$. This is a contradiction to (\ref{2.31}).
\end{proof}

In the following, we show that the solution of \eqref{1.2} is strongly uniform persistence.
\begin{lemma}\label{L9}
Assume that (J), (f) and (H) are satisfied. Then for $c>c^*$,
$\liminf_{\xi\rightarrow\infty}S(\xi)>\sigma$, $\liminf_{\xi\rightarrow\infty}P(\xi)>0$ and
$\liminf_{\xi\rightarrow\infty}I(\xi)>0$.
\end{lemma}
\begin{proof}
 Note  $\liminf_{\xi\rightarrow\infty}S(\xi)\geq \sigma.$  We  assume  $\liminf_{\xi\rightarrow\infty}S(\xi)=\sigma$.
Find a sequence $\{\xi_n\}$ with $\lim_{n\rightarrow\infty}\xi_n=\infty$ such that
$\lim_{n\rightarrow\infty}S(\xi_n)=\sigma$. Up to extraction a subsequence, $S(\xi+\xi_n)\to \widetilde{S}(\xi), P(\xi+\xi_n)\to \widetilde{P}(\xi), I(\xi+\xi_n)\to \widetilde{I}(\xi)$ in $C^{1}_{loc}(\mathbb{R})$ as $n\rightarrow\infty$. Furthermore, the nonnegative functions $\widetilde{S}(\cdot), \widetilde{P}(\cdot), \widetilde{I}(\cdot)\in C^{1}_{loc}(\mathbb{R})$.
Substituting $\xi+\xi_n$ in S-equation of \eqref{2.1} and letting $n \rightarrow\infty$, we obtain
\begin{equation}\label{2.35}
c\widetilde{S}'(\xi)=d_1\mathcal {A}[\tilde{S}](\xi)+ b_1[K_1-\widetilde{S}(\xi)]- \widetilde{S}(\xi)\widetilde{I}(\xi)+\sigma \widetilde{I}(\xi).
\end{equation}
 Letting $\xi=0$ in \eqref{2.35} and noting $\widetilde{S}(0)=\sigma, \widetilde{S}'(0)= 0$, then
\begin{equation*}
0=c\widetilde{S}'(0)=d_1\mathcal {A}[\tilde{S}](0)+ b_1[K_1-\sigma].
\end{equation*}
Hence, $b_1[K_1-\sigma]=-d_1\mathcal {A}[\tilde{S}](0)\leq 0$, which contradicts $b_1[K_1-\sigma]>0$. Thus,
$\liminf_{\xi\rightarrow\infty}S(\xi)> \sigma$. Similarly, we can prove  $\liminf_{\xi\rightarrow\infty}P(\xi)>0$.
By the positivity of $I$ on $\mathbb{R}$ and Lemma \ref{L8}, we immediately obtain $\liminf_{\xi\rightarrow\infty}I(\xi)>0$.

\end{proof}

\begin{lemma}\label{L10}
Assume that (J), (f) and (H) are satisfied.  Then for $c>c^*$, $\lim_{\xi\rightarrow\infty}(S,P,I)(\xi)=(S^*,P^*,I^*)$.
\end{lemma}
\begin{proof}
Set $\tilde{g}(x)=x-1-\ln x$ and $\tilde{f}(x)=x-\int_{S^*}^{x}\frac{S^*-\sigma}{\tau-\sigma}\dif \tau$. Define a Lyapunov functional as follows
\begin{equation}\label{2.38}
\begin{aligned}
&\mathcal {L}_1(S,P,I)(\xi)=c\int_{S^*}^{S(\xi)}\frac{\tau-S^*}{\tau-\sigma}\dif \tau
+\frac{c}{\gamma_1}I^*\tilde{g}\left(\frac{I(\xi)}{I^*}\right)+\frac{c\gamma_2}{\gamma_1}P^*\tilde{g}\left(\frac{P(\xi)}{P^*}\right),\\
&\mathcal {L}_2(S,I)(\xi)=S^*I^*\int_{0}^hf(\tau)\int_{\xi-c\tau}^{\xi}\tilde{g}\left(\frac{S(\eta)I(\eta)}{S^*I^*}\right)\dif \eta\dif \tau,\\
&\mathcal {L}_3(I)(\xi)=\sum_{i=1}^{N_1}J(i)\int_{\xi-i}^{\xi}\tilde{g}\left(\frac{I(\eta)}{I^*}\right)\dif \eta-\sum_{i=1}^{N_1}J(i)\int_{\xi}^{\xi+i}\tilde{g}\left(\frac{I(\eta)}{I^*}\right)\dif \eta,\\
&\mathcal {L}_4(S)(\xi)=\sum_{i=1}^{N_1}J(i)\int_{\xi-i}^{\xi}\tilde{f}(S(\eta))\dif \eta-\sum_{i=1}^{N_1}J(i)\int_{\xi}^{\xi+i}\tilde{f}(S(\eta))\dif \eta,\\
&\mathcal {L}_5(P)(\xi)=\sum_{i=1}^{N_1}J(i)\int_{\xi-i}^{\xi}\tilde{g}\left(\frac{P(\eta)}{P^*}\right)\dif \eta-\sum_{i=1}^{N_1}J(i)\int_{\xi}^{\xi+i}\tilde{g}\left(\frac{P(\eta)}{P^*}\right)\dif \eta.\\
\end{aligned}
\end{equation}
Notice that $S,P,I$ are all bounded and Lemma \ref{L9} gives  $S,P,I$  are all strongly uniform persistence. Thus $\mathcal {L}_i(\cdot)(\xi), i=1,2,3,4,5$ are all well-defined and bounded from below for $\xi$ sufficiently large.
Now let
\begin{equation}\label{2.39}
\mathcal {L}(S,P,I)(\xi)=\mathcal {L}_1(S,P,I)(\xi)+\mathcal {L}_2(S,I)(\xi)+\frac{d_3}{\gamma_1}I^*\mathcal {L}_3(I)(\xi)+d_1\mathcal {L}_4(S)(\xi)+\frac{d_2\gamma_2}{\gamma_1}P^*\mathcal {L}_5(P)(\xi).
\end{equation}
Some calculations show that
\begin{equation*}
\begin{aligned}
\frac{\dif\mathcal {L}_1(S,P,I)(\xi)}{\dif \xi}
&=\frac{S(\xi)-S^*}{S(\xi)-\sigma}cS'(\xi)+\frac{1}{\gamma_1}\left[1-\frac{I^*}{I(\xi)}\right]cI'(\xi)
+\frac{\gamma_2}{\gamma_1}\left[1-\frac{P^*}{P(\xi)}\right]cP'(\xi)\\
&=\frac{S(\xi)-S^*}{S(\xi)-\sigma}d_1\mathcal {A}[S](\xi)+\frac{d_3}{\gamma_1}\left[1-\frac{I^*}{I(\xi)}\right]\mathcal {A}[I](\xi)+\frac{d_2\gamma_2}{\gamma_1}\left[1-\frac{P^*}{P(\xi)}\right]\mathcal {A}[P](\xi)\\
&\qquad+G_1(S,P,I)(\xi),\\
\end{aligned}
\end{equation*}
where
\begin{equation*}
\begin{aligned}
G_1(S,P,I)(\xi)&=\frac{S(\xi)-S^*}{S(\xi)-\sigma}\big[b_1(K_1-S(\xi))- (S(\xi)-\sigma)I(\xi)\big]\\
&+\frac{\gamma_2}{\gamma_1}\left[1-\frac{P^*}{P(\xi)}\right]\left[b_2(K_2-P(\xi))- P(\xi)I(\xi)\right]\\
&+\frac{1}{\gamma_1}\left[1-\frac{I^*}{I(\xi)}\right]\left[\gamma_1 \int_{0}^hf(\tau)S(\xi-c\tau)I(\xi-c\tau)\dif \tau-\delta I(\xi)+\gamma_2P(\xi)I(\xi)\right].
\end{aligned}
\end{equation*}
By using $b_1K_1=b_1S^*+S^*I^*-\sigma I^*$, $b_2(K_2-P^*)=P^*I^*$ and $\gamma_1 S^*I^*+\gamma_2P^*I^*=\delta I^*$, we obtain
\begin{equation}\label{2.40}
\begin{aligned}
G_1(S,P,I)(\xi)&=-b_1\frac{[S(\xi)-S^*]^2}{S(\xi)-\sigma}+[S(\xi)-S^*][I^*-I(\xi)]-\frac{I^*[S(\xi)-S^*]^2}{S(\xi)-\sigma}\\
&-\frac{b_2\gamma_2}{\gamma_1P(\xi)}[P(\xi)-P^*]^2+\frac{\gamma_2}{\gamma_1}\left[1-\frac{P^*}{P(\xi)}\right][P^*I^*-P(\xi)I(\xi)]\\
&+\frac{1}{\gamma_1}[\gamma_2P(\xi)-\gamma_1S^*-\gamma_2P^*][I(\xi)-I^*]\\
&+\int_{0}^hf(\tau)S(\xi-c\tau)I(\xi-c\tau)\dif \tau-\frac{I^*}{I(\xi)}\int_{0}^hf(\tau)S(\xi-c\tau)I(\xi-c\tau)\dif \tau.
\end{aligned}
\end{equation}
On the other hand,
\begin{equation}\label{2.41}
\frac{\dif\mathcal {L}_2(S,I)(\xi)}{\dif \xi}=S^*I^*\int_{0}^hf(\tau)\left[\frac{S(\xi)I(\xi)}{S^*I^*}-\frac{S(\xi-c\tau)I(\xi-c\tau)}{S^*I^*}
+\ln\frac{S(\xi-c\tau)I(\xi-c\tau)}{S(\xi)I(\xi)}\right]\dif \tau.
\end{equation}
By  \eqref{2.40} and  \eqref{2.41}, we have
\begin{equation}\label{2.42}
\begin{aligned}
&G_1(S,P,I)(\xi)+\frac{\dif\mathcal {L}_2(S,I)(\xi)}{\dif \xi}\\
=&-(b_1+I^*)\frac{[S(\xi)-S^*]^2}{S(\xi)-\sigma}-\frac{b_2\gamma_2+\gamma_2I^*}{\gamma_1P(\xi)}[P(\xi)-P^*]^2+S(\xi)I^*\\
&-S^*I^*\int_{0}^hf(\tau)\left[\frac{S(\xi-c\tau)I(\xi-c\tau)}{S^*I(\xi)}-\ln\frac{S(\xi-c\tau)I(\xi-c\tau)}{S(\xi)I(\xi)}\right]\dif \tau\\
=&-(b_1+I^*)\frac{[S(\xi)-S^*]^2}{S(\xi)-\sigma}-\frac{b_2\gamma_2+\gamma_2I^*}{\gamma_1P(\xi)}[P(\xi)-P^*]^2
+\frac{I^*[S(\xi)-S^*]^2}{S(\xi)}\\
&-S^*I^*\int_{0}^hf(\tau)
\left[\frac{S^*}{S(\xi)}-2+\frac{S(\xi-c\tau)I(\xi-c\tau)}{S^*I(\xi)}-\ln\frac{S(\xi-c\tau)I(\xi-c\tau)}{S(\xi)I(\xi)}\right]\dif \tau\\
=&-b_1\frac{[S(\xi)-S^*]^2}{S(\xi)-\sigma}-\frac{b_2\gamma_2+\gamma_2I^*}{\gamma_1P(\xi)}[P(\xi)-P^*]^2-\frac{\sigma I^*[S(\xi)-S^*]^2}{(S(\xi)-\sigma)S(\xi)}\\
&-S^*I^*\int_{0}^hf(\tau)
\left[\tilde{g}\left(\frac{S^*}{S(\xi)}\right)+\tilde{g}\left(\frac{S(\xi-c\tau)I(\xi-c\tau)}{S^*I(\xi)}\right)\right]\dif \tau\leq 0.
\end{aligned}
\end{equation}
In addition, we get
\begin{equation}\label{2.43}
\frac{\dif\mathcal {L}_3(I)(\xi)}{\dif \xi}=\sum_{i=1}^{N_1}J(i)\left[2\tilde{g}\left(\frac{I(\xi)}{I^*}\right)-\tilde{g}\left(\frac{I(\xi-i)}{I^*}\right)
-\tilde{g}\left(\frac{I(\xi+i)}{I^*}\right)\right],
\end{equation}
\begin{equation}\label{2.44}
\frac{\dif\mathcal {L}_4(S)(\xi)}{\dif \xi}=\sum_{i=1}^{N_1}J(i)\big[2\tilde{f}(S(\xi))-\tilde{f}(S(\xi-i))-\tilde{f}(S(\xi+i))\big],
\end{equation}
\begin{equation}\label{2.44a}
\frac{\dif\mathcal {L}_5(P)(\xi)}{\dif \xi}=\sum_{i=1}^{N_1}J(i)\left[2\tilde{g}\left(\frac{P(\xi)}{P^*}\right)-\tilde{g}\left(\frac{P(\xi-i)}{P^*}\right)
-\tilde{g}\left(\frac{P(\xi+i)}{P^*}\right)\right].
\end{equation}
It follows from \eqref{2.43}  that
\begin{equation}\label{2.45}
\begin{aligned}
&\frac{d_3}{\gamma_1}\left[1-\frac{I^*}{I(\xi)}\right]\mathcal {A}[I](\xi)+\frac{d_3}{\gamma_1}I^*\frac{\dif\mathcal {L}_3(I)(\xi)}{\dif \xi}\\
=&-\frac{d_3}{\gamma_1}I^*\sum_{i=1}^{N_1}J(i)\left[\frac{I(\xi-i)}{I(\xi)}-1-\ln\frac{I(\xi-i)}{I(\xi)}
+\frac{I(\xi+i)}{I(\xi)}-1-\ln\frac{I(\xi+i)}{I(\xi)}\right]\\
=&-\frac{d_3}{\gamma_1}I^*\sum_{i=1}^{N_1}J(i)\left[\tilde{g}\left(\frac{I(\xi-i)}{I(\xi)}\right)+\tilde{g}\left(\frac{I(\xi+i)}{I(\xi)}\right)\right]\leq 0.
\end{aligned}
\end{equation}
It follows from \eqref{2.44a}  that
\begin{equation}\label{2.45a}
\begin{aligned}
&\frac{d_2\gamma_2}{\gamma_1}\left[1-\frac{P^*}{P(\xi)}\right]\mathcal {A}[P](\xi)
+\frac{d_2\gamma_2}{\gamma_1}P^*\frac{\dif\mathcal {L}_5(P)(\xi)}{\dif \xi}\\
=&-\frac{d_2\gamma_2}{\gamma_1}P^*\sum_{i=1}^{N_1}J(i)\left[\frac{P(\xi-i)}{P(\xi)}-1-\ln\frac{P(\xi-i)}{P(\xi)}
+\frac{P(\xi+i)}{P(\xi)}-1-\ln\frac{P(\xi+i)}{P(\xi)}\right]\\
=&-\frac{d_2\gamma_2}{\gamma_1}P^*\sum_{i=1}^{N_1}J(i)\left[\tilde{g}\left(\frac{P(\xi-i)}{P(\xi)}\right)
+\tilde{g}\left(\frac{P(\xi+i)}{P(\xi)}\right)\right]\leq 0.
\end{aligned}
\end{equation}
Define $R(x)=x-S(\xi)-\int_{S(\xi)}^{x}\frac{S(\xi)-\sigma}{\tau-\sigma}\dif \tau$. Note that
\begin{equation*}
\left\{\begin{array}{l}
R(S(\xi))=0,\\
R'(x)>0, \text{ if  } x>S(\xi)>\sigma,\\
R'(x)<0, \text{ if  } \sigma<x<S(\xi)\\
\end{array}\right.\Rightarrow R(x)\geq 0, \forall\; x>\sigma \text{ and } R(x)=0\Leftrightarrow x=S(\xi).
\end{equation*}
By \eqref{2.44} and the properties of $R(\cdot)$, we see that
\begin{equation}\label{2.46}
\begin{aligned}
&d_1\frac{S(\xi)-S^*}{S(\xi)-\sigma}\mathcal {A}[S](\xi)+d_1\frac{\dif\mathcal {L}_4(S)(\xi)}{\dif \xi}\\
=&-\frac{d_1(S^*-\sigma)}{S(\xi)-\sigma}\sum_{i=1}^{N_1}J(i)\left[S(\xi-i)+S(\xi+i)-2S(\xi)-\int_{S(\xi)}^{S(\xi-i)}\frac{S(\xi)-\sigma}{\tau-\sigma}\dif \tau-\int_{S(\xi)}^{S(\xi+i)}\frac{S(\xi)-\sigma}{\tau-\sigma}\dif \tau\right]\\
=&-\frac{d_1(S^*-\sigma)}{S(\xi)-\sigma}\sum_{i=1}^{N_1}J(i)\left[R(S(\xi-i))+R(S(\xi+i))\right]\leq 0.
\end{aligned}
\end{equation}
By \eqref{2.38}, \eqref{2.39}, \eqref{2.42}, \eqref{2.45}, \eqref{2.45a} and \eqref{2.46}, we reach
\begin{equation*}
\begin{aligned}
\frac{\dif\mathcal {L}(S,P,I)(\xi)}{\dif \xi}=&G_1(S,P,I)(\xi)+\frac{\dif\mathcal {L}_2(S,I)(\xi)}{\dif \xi}\\
&+d_1\frac{S(\xi)-S^*}{S(\xi)-\sigma}\mathcal {A}[S](\xi)+d_1\frac{\dif\mathcal {L}_4(S)(\xi)}{\dif \xi}\\
&+\frac{d_2\gamma_2}{\gamma_1}\left[1-\frac{P^*}{P(\xi)}\right]\mathcal {A}[P](\xi)+\frac{d_2\gamma_2}{\gamma_1}P^*\frac{\dif\mathcal {L}_5(P)(\xi)}{\dif \xi}\\
&+\frac{d_3}{\gamma_1}\left[1-\frac{I^*}{I(\xi)}\right]\mathcal {A}[I](\xi)+\frac{d_3}{\gamma_1}I^*\frac{\dif\mathcal {L}_3(I)(\xi)}{\dif \xi}\leq 0.
\end{aligned}
\end{equation*}
Note that $\frac{\dif\mathcal {L}(S,P,I)(\xi)}{\dif \xi}=0$ implies $\tilde{g}(\frac{P(\xi\pm i)}{P(\xi)})=0$, $\tilde{g}(\frac{I(\xi\pm i)}{I(\xi)})=0$ and  $R(S(\xi\pm i))=0$. Thus $S(\xi), P(\xi), I(\xi)$
must be all constant functions. Furthermore, $G_1(S,P,I)(\xi)+\frac{\dif\mathcal {L}_2(S,I)(\xi)}{\dif \xi}=0$
results in $S(\xi)\equiv S^*, P(\xi)\equiv P^*$ and $I(\xi)\equiv I^*$. Note that the maximum invariant set of
$\{(S,P,I): \frac{\dif\mathcal {L}(S,P,I)(\xi)}{\dif \xi}=0\}$ consists of only one equilibrium $(S^*,P^*,I^*)$.
By LaSalle's invariance principle, $\lim_{\xi\rightarrow\infty}(S,P,I)(\xi)=(S^*,P^*,I^*)$.

\end{proof}

\begin{theorem}\label{T11}
Assume that (J), (f) and (H) hold. For $c>c^*$, the system \eqref{2.1} has a solution $(S,P,I)$ such that $\sigma<S<K_1$, $0<P<K_2$ and $I>0$ on
$\mathbb{R}$. Moreover, the boundary conditions \eqref{2.2a} hold.
\end{theorem}

{\section{Existence of  epidemic wave for $c=c^*$}}

We will use the limiting argument to establish the existence of epidemic waves for the critical case. Notice the lower solution of \eqref{2.1}
will  be degenerate at $c=c^*$, we need some proper shifts to avoid the problem.

\begin{theorem}\label{T12}
Assume that (J), (f) and (H) hold. For $c=c^*$, the system \eqref{2.1} has a solution $(S,P,I)$ such that $\sigma<S<K_1$, $0<P<K_2$ and $I>0$ on
$\mathbb{R}$. Moreover, the boundary conditions \eqref{2.2a} hold.
\end{theorem}
\begin{proof}
Pick a sequence $\{c_n\}$ with $c^*<c_n<c^*+\frac{1}{n}$ for each $n\in \mathbb{Z}_+$.
Let $\{(S_n,P_n,I_n)\}$ be the solution sequence of \eqref{2.1} corresponding to wave speeds  $\{c_n\}$. By Theorem \ref{T4}, we see that
$\sigma<S_n<K_1$, $0<P_n<K_2$ and $I_n>0$ on $\mathbb{R}$. Moreover, according to Lemma \ref{L6}, the sequence $\{I_n\}$ is also uniformly bounded on $\mathbb{R}$.

We next show that the solution sequence cannot converge to zero uniformly as $n\to\infty$.

\textbf{Step 1.}  Prove $\liminf_{n\to\infty}\|I_n(\cdot)\|_{L^{\infty}(\mathbb{R})}>0$.

If not, up to extraction of a subsequence, we may as well suppose that $\|I_n(\cdot)\|_{L^{\infty}(\mathbb{R})}\to 0$ as $n\to\infty$.
Let $M_{n}=\|I_n(\cdot)\|_{L^{\infty}(\mathbb{R})}$. Then we can find some small $\varsigma_0$ and some large $n_0$ such that $M_n\leq \varsigma_0$ for all $n\geq n_0$. Let $c_1=c^*$ and $c_2=c^*+1$ in Lemma \ref{L8}. Then $I_n'(\xi)\geq 0$ for $n\geq n_0$. On the other side, $I_n$ is  uniformly bounded and positive. Thus $\lim_{\xi\to\infty}I_n(\xi)$ exists and is positive for $n\geq n_0$. This is a contradiction.
By Step 1 and the positivity of $I_n(\cdot)$, we can assume  $\inf_{n}\|I_n(\cdot)\|_{L^{\infty}(\mathbb{R})}>0$ without loss of generality.

Note that $\sigma<S_n(\xi)<K_1$, $0<P_n(\xi)<K_2$ for all $n\in\mathbb{Z}_+$ and $\xi\in\mathbb{R}$. Thus, $\inf_{\mathbb{R}}S_n(\xi)\leq S_n(\xi)<K_1$ and $\inf_{\mathbb{R}}P_n(\xi)\leq P_n(\xi)<K_2$. It implies that there is $\varsigma_n$ such that  $\inf_{\mathbb{R}}S_n(\xi)\leq K_1-\varsigma_n$
and $\inf_{\mathbb{R}}P_n(\xi)\leq K_2-\varsigma_n$. In the following, we show that this $\varsigma_n$  can be independent of $n$.

\textbf{Step 2.} Prove $\inf_{\mathbb{R}}S_n(\xi)\leq K_1-\tilde{\varsigma}$ and  $\inf_{\mathbb{R}} P_n(\xi)\leq K_2-\tilde{\varsigma}$ for some   $0<\tilde{\varsigma}<\min\{K_1,K_2\}$.

For otherwise, up to extraction a subsequence, assume that $a_n:=\inf_{\mathbb{R}}S_n(\xi)\to K_1$ as $n\to\infty$. From
$$K_1=\lim_{n\to\infty}\inf_{\mathbb{R}}S_n(\xi)\leq \lim_{n\to\infty}\sup_{\mathbb{R}}S_n(\xi)\leq K_1,$$
we have $\lim_{n\to\infty}S_n(\xi)=K_1$ uniformly for $\xi\in \mathbb{R}$. Remember
\begin{equation}\label{4.2}
c_nS_n'(\xi)=d_1\mathcal {A}[S_n](\xi)+ b_1[K_1-S_n(\xi)]-[S_n(\xi)-\sigma]I_n(\xi).
\end{equation}
Letting $n\to\infty$ in  \eqref{4.2} and noting $S_n'(\xi)$ is uniformly convergent  in $C_{\text{loc}}(\mathbb{R})$, then $S_n'(\xi)\to 0$ as $n\to\infty$. Hence, $I_n(\xi)\to 0$ as $n\to\infty$ uniformly for $\xi\in\mathbb{R}$. This leads to $\sup_{\mathbb{R}}I_n(\xi)\to 0$ as $n\to\infty$.
It contradicts with $\inf_{n}\|I_n(\cdot)\|_{L^{\infty}(\mathbb{R})}>0$. Similarly, we can prove $\inf_{\mathbb{R}} P_n(\xi)\leq K_2-\tilde{\varsigma}$ for some   $0<\tilde{\varsigma}<K_2$.

 Note that $\lim_{\xi\rightarrow-\infty}S_n(\xi)=K_1$ and $\lim_{\xi\rightarrow-\infty}P_n(\xi)=K_2$. By Step 2, choose two constants $\varsigma_1,\varsigma_2\in(0,\tilde{\varsigma})$. By translation along to $\xi$-axis, we may assume that $S_n(0)=K_1-\varsigma_1, P_n(0)=K_2-\varsigma_2$, and $S_n(\xi)\geq K_1-\varsigma_1, P_n(\xi)\geq K_2-\varsigma_2$ for all $\xi\leq 0$.
As discussions above, there are some functions $S_{\star}(\xi),P_{\star}(\xi),I_{\star}(\xi)$ exist such that
$(S_n(\xi),P_n(\xi),I_n(\xi))\to (S_{\star}(\xi),P_{\star}(\xi),I_{\star}(\xi))$ in $\left[C^1_{\text{loc}}(\mathbb{R})\right]^3$  as  $n\to\infty$.
Also, $(S_{\star}(\xi),P_{\star}(\xi),I_{\star}(\xi))$ satisfies
\begin{equation}\label{4.3}
\left\{\begin{array}{l}
c^*S_{\star}'(\xi)=d_1\mathcal {A}[S_{\star}](\xi)+ b_1[K_1-S_{\star}(\xi)]- S_{\star}(\xi)I_{\star}(\xi)+\sigma I_{\star}(\xi),\\
c^*P_{\star}'(\xi)=d_2\mathcal {A}[P_{\star}](\xi)+ b_2[K_2-P_{\star}(\xi)]- P_{\star}(\xi)I_{\star}(\xi),\\
c^*I_{\star}'(\xi)=d_3\mathcal {A}[I_{\star}](\xi)+\gamma_1 \displaystyle{\int_{0}^hf(\tau)S_{\star}(\xi-c\tau)I_{\star}(\xi-c\tau)\dif \tau}-\delta I_{\star}(\xi)+\gamma_2P_{\star}(\xi)I_{\star}(\xi).
\end{array}\right.
\end{equation}
It is easy to prove $\sigma< S_{\star}< K_1$, $0< P_{\star}<K_2$ and $I_{\star}>0$ on $\mathbb{R}$.

\textbf{Step 3.} Prove $\lim_{\xi\rightarrow-\infty}(S_{\star}(\xi),P_{\star}(\xi),I_{\star}(\xi))=(K_1,K_2,0)$.

Note that $K_1>S_{\star}(z)\geq K_1-\varsigma_1$ and  $K_2>P_{\star}(z)\geq K_2-\varsigma_2$ for all $z\leq 0$.
Then $\inf_{z\in(-\infty,0]}S_{\star}(z)\to K_1$ and $\inf_{z\in(-\infty,0]}P_{\star}(z)\to K_2$ as
$\tilde{\varsigma}\to 0$. Notice that  $\inf_{z\in(-\infty,0]}S_{\star}(z)\leq \inf_{z\in(-\infty,-c^*\tau]}S_{\star}(z)\leq \inf_{z\in(-\infty,-c^*h]}S_{\star}(z)$. Then for any $\tau\in[0,h]$, $\inf_{z\in(-\infty,-c^*\tau]}S_{\star}(z)\to K_1$ as
$\tilde{\varsigma}\to 0$. Pick a sufficiently small $\tilde{\varsigma}$ such that
$$\sup_{z\in(-\infty,0]}[K_1-S_{\star}(z-c^*\tau)]=K_1-\inf_{z\in(-\infty,0]}S_{\star}(z-c^*\tau)=:\eta_1<\frac{\triangle_K(0,c^*)}{4\gamma_1}$$ and
$$\sup_{z\in(-\infty,0]}[K_2-P_{\star}(z)]=K_2-\inf_{z\in(-\infty,0]}P_{\star}(z)=:\eta_2<\frac{\triangle_K(0,c^*)}{4\gamma_2}.$$
For $\lambda<0$, we can define $\mathscr{N}[I_{\star}](\lambda)=\int_{-\infty}^{0}I_{\star}(z)\me^{-\lambda z}\dif z$. Indeed,
 $\int_{-\infty}^{0}I_{\star}(z)\me^{-\lambda z}\dif z=\int^{\infty}_{0}I_{\star}(-s)\me^{\lambda s}\dif s$ and $\lim_{s\rightarrow\infty}I_{\star}(-s)\me^{\lambda s}s^2=0$, by Cauchy's integral converge principle, then $\int^{\infty}_{0}I_{\star}(-s)\me^{\lambda s}\dif s$ converges. The uniform boundedness of $I'_{\star}(\xi)$ shows the uniform continuity of  $I_{\star}(\xi)$. If $\int_{-\infty}^{0}I_{\star}(\xi)\dif \xi<\infty$, then $\lim_{\xi\rightarrow-\infty}I_{\star}(\xi)=0$ and we have  achieved the aim. So we might as well suppose that $\int_{-\infty}^{0}I_{\star}(\xi)\dif \xi\to\infty$.
 Taking the negative Laplace transform of $I_{\star}$-equation in \eqref{4.3}, we get
\begin{equation*}\label{4.4}
c^*I_{\star}(0)+B_1(\lambda)+B_2(\lambda)=\triangle_K(\lambda,c^*)\mathscr{N}[I_{\star}](\lambda),
\end{equation*}
wherein
\begin{equation*}
\begin{aligned}
B_1(\lambda)&=d_3\sum_{i=1}^{N_1}J(i)\left[\int_{-i}^{0}I_{\star}(z)\me^{-\lambda (z+i)}\dif z-\int^{i}_{0}I_{\star}(z)\me^{-\lambda (z-i)}\dif z\right]\\
&\quad\quad+\gamma_1 K_1\int_{0}^hf(\tau)\int_{-c^*\tau}^{0}I_{\star}(z)\me^{-\lambda (z+c^*\tau)}\dif z\dif \tau,\\
B_2(\lambda)&=\gamma_1\int_{-\infty}^{0}\int_{0}^hf(\tau)[K_1-S_{\star}(z-c^*\tau)]I_{\star}(z-c^*\tau)\me^{-\lambda z}\dif \tau\dif z\\
&\quad\quad+\gamma_2\int_{-\infty}^{0}[K_2-P_{\star}(z)]I_{\star}(z)\me^{-\lambda z}\dif z.
\end{aligned}
\end{equation*}
Since
\begin{equation*}
\begin{aligned}
B_2(\lambda)&\leq\gamma_1\int_{-\infty}^{0}\int_{0}^hf(\tau)\sup_{z\in(-\infty,0]}
[K_1-S_{\star}(z-c^*\tau)]I_{\star}(z-c^*\tau)\me^{-\lambda z}\dif \tau\dif z\\
&\quad\quad+\gamma_2\int_{-\infty}^{0}\sup_{z\in(-\infty,0]}[K_2-P_{\star}(z)]I_{\star}(z)\me^{-\lambda z}\dif z\\
&=\gamma_1\eta_1\int_{-\infty}^{0}\int_{0}^hf(\tau)I_{\star}(z-c^*\tau)\me^{-\lambda z}\dif \tau\dif z+\gamma_2\eta_2\int_{-\infty}^{0}I_{\star}(z)\me^{-\lambda z}\dif z\\
&=\gamma_1\eta_1\left[\mathscr{N}[I_{\star}](\lambda)\int_{0}^hf(\tau)\me^{-\lambda c^*\tau}\dif \tau-\int_{0}^h\int_{-c^*\tau}^{0}f(\tau)\me^{-\lambda (z+c^*\tau)}I_{\star}(z)\dif z\dif \tau\right]+\gamma_2\eta_2\mathscr{N}[I_{\star}](\lambda),
\end{aligned}
\end{equation*}
we get
\begin{equation*}
\begin{aligned}
B_1(\lambda)+B_2(\lambda)&\leq d_3\sum_{i=1}^{N_1}J(i)\left[\int_{-i}^{0}I_{\star}(z)\me^{-\lambda (z+i)}\dif z-\int^{i}_{0}I_{\star}(z)\me^{-\lambda (z-i)}\dif z\right]\\
&+\gamma_1 (K_1-\eta_1)\int_{0}^hf(\tau)\int_{-c^*\tau}^{0}I_{\star}(z)\me^{-\lambda (z+c^*\tau)}\dif z\dif \tau\\
&\quad\quad+\gamma_1\eta_1\mathscr{N}[I_{\star}](\lambda)\int_{0}^hf(\tau)\me^{-\lambda c^*\tau}\dif \tau+\gamma_2\eta_2\mathscr{N}[I_{\star}](\lambda).
\end{aligned}
\end{equation*}
At last, we have
\begin{equation}\label{4.5}
\begin{aligned}
&\left[\triangle_K(\lambda,c^*)-\gamma_1\eta_1\int_{0}^hf(\tau)\me^{-\lambda c^*\tau}\dif \tau-\gamma_2\eta_2\right]\mathscr{N}[I_{\star}](\lambda)\\
&\leq d_3\sum_{i=1}^{N_1}J(i)\left[\int_{-i}^{0}I_{\star}(z)\me^{-\lambda (z+i)}\dif z-\int^{i}_{0}I_{\star}(z)\me^{-\lambda (z-i)}\dif z\right]\\
&+\gamma_1 (K_1-\eta_1)\int_{0}^hf(\tau)\int_{-c^*\tau}^{0}I_{\star}(z)\me^{-\lambda (z+c^*\tau)}\dif z\dif \tau+c^*I_{\star}(0).
\end{aligned}
\end{equation}
Since $\triangle_K(0,c^*)-\gamma_1\eta_1-\gamma_2\eta_2>\frac{1}{2}\triangle_K(0,c^*)>0$, setting $\lambda\to 0$ in \eqref{4.5}, we have a contradiction. In fact, in this situation, the left-hand side of \eqref{4.5} is unbounded while the right-hand side of \eqref{4.5} is bounded. Therefore,   $\lim_{\xi\rightarrow-\infty}I_{\star}(\xi)=0$.

We now prove $\lim_{\xi\rightarrow-\infty}S_{\star}(\xi)=K_1$ and $\lim_{\xi\rightarrow-\infty}P_{\star}(\xi)=K_2$. Assume
$$\underline{S}_{\star}:=\liminf_{\xi\rightarrow-\infty}S_{\star}(\xi)<\limsup_{\xi\rightarrow-\infty}S_{\star}(\xi)=:\bar{S}_{\star}.$$
By fluctuation lemma, there are two sequences $\{x_k\}$ and $\{y_k\}$ with $x_k, y_k\to -\infty$ as $k\to\infty$ such that
$$S_{\star}(x_k)\to \underline{S}_{\star}\text { as  }k\to\infty, \qquad \underline{S}_{\star}'(x_k)=0,$$
$$S_{\star}(y_k)\to \bar{S}_{\star}\text { as  }k\to\infty, \qquad \bar{S}_{\star}'(y_k)=0.$$
We follow from $S_{\star}$-equation in \eqref{4.3} that
$0= b_1[K_1-\underline{S}_{\star}]$ and $0= b_1[K_1-\bar{S}_{\star}]$.
Thus $\underline{S}_{\star}=K_1=\bar{S}_{\star}$. It implies that $\lim_{\xi\rightarrow-\infty}S_{\star}(\xi)=K_1$. In a similar way, we can show that
 $\lim_{\xi\rightarrow-\infty}P_{\star}(\xi)=K_2$.

 Note that the Lyapunov functional still holds for the case $c=c^*$. Hence,  $\lim_{\xi\rightarrow\infty}(S,P,I)(\xi)=(S^*,P^*,I^*)$.

\end{proof}

{\section{Non-existence of  epidemic wave for $0<c<c^*$}}
\begin{theorem}\label{T13}
Assume that (J), (f) and (H) hold. For $0<c<c^*$, the system \eqref{2.1} has no positive solution $(S,P,I)$ such that $\sigma<S<K_1$ and $P, I>0$ on
$\mathbb{R}$ and $\lim_{\xi\rightarrow-\infty}(S,P,I)(\xi)=(K_1,K_2,0)$.
\end{theorem}
\begin{proof}
By Lemma \ref{L7}, for any $\lambda\in \mathbb{C}$ with $0<\text{Re}\lambda<\rho$, the bilateral Laplace transform of $I$, that is,   $\mathscr{L}[I](\lambda)=\int_{-\infty}^{\infty}I(z)\me^{-\lambda z}\dif z$ is well-defined.
Performing the bilateral Laplace transform of $I$ to the I-equation of \eqref{2.1}, we have
\begin{equation}\label{5.1}
\begin{aligned}
&\triangle_K(\lambda,c)\mathscr{L}[I](\lambda)\\
=&\gamma_1\int_{0}^h f(\tau)\me^{-\lambda c\tau}\dif \tau\int_{-\infty}^{\infty} \big[K_1-S(z)\big]I(z)\me^{-\lambda z}\dif z+\gamma_2\int_{-\infty}^{\infty} \big[K_2-P(z)\big]I(z)\me^{-\lambda z}\dif z.
\end{aligned}
\end{equation}
Assume such solution did exist.  Note $\lim_{\xi\rightarrow-\infty}(S,P,I)(\xi)=(K_1,K_2,0)$.  For any $\epsilon_0>0$ sufficiently small, choose $z^*$ so large that $K_1-S(z)<\epsilon_0, K_2-P(z)<\epsilon_0$ and $I(z)=\frac{1}{2}\epsilon_0$ for $z<-z^*$. Then for $z<-z^*$,
\begin{equation*}
\big[K_1-S(z)\big]I(z)\me^{-2\rho z}<\frac{\epsilon_0}{I(z)}\big[I(z)\me^{-\rho z}\big]^2=2\big[I(z)\me^{-\rho z}\big]^2\leq2M_2^2.
\end{equation*}
Similarly, $\big[K_2-P(z)\big]I(z)\me^{-2\rho z}<2M_2^2$.
Hence the right-hand side of  \eqref{5.1} is well-defined for $0<\text{Re}\lambda<2\rho$.  Lemma \ref{L1} shows that $\triangle_K(\lambda,c)>0$
for $0<c<c^*$ and $\lambda>0$. Thus, $\mathscr{L}[I](\lambda)=\int_{-\infty}^{\infty}I(z)\me^{-\lambda z}\dif z$ can be extended to $0<\text{Re}\lambda<2\rho$. In the following we prove $\mathscr{L}[I](\lambda)$ is well-defined for all $\text{Re}\lambda>0$.
In view of the property of Laplace transform, we see that either there is some number $\varrho$ such that  $\mathscr{L}[I](\lambda)$ is
analytic in the strip $0<\text{Re}\lambda<\varrho$ and has a singularity at $\lambda=\varrho$ or $\mathscr{L}[I](\lambda)$ is well-defined for all $\text{Re}\lambda>0$. If $2\rho<\varrho$, then for any $0<\theta<1$ and $\theta^*:=2\rho+\theta(\varrho-2\rho)<\varrho$, we have
$I(z)\me^{-\theta^* z}\leq M_2$ for all $z\in \mathbb{R}$. Similarly, $\big[K_1-S(z)\big]I(z)\me^{-2\theta^* z} \leq2M_2^2$ and
$\big[K_2-P(z)\big]I(z)\me^{-2\theta^* z} \leq2M_2^2$. Then the right-hand side of  \eqref{5.1} is well-defined for $0<\text{Re}\lambda<2\theta^*$. If $\max\{0,\frac{\varrho-\theta^*-2\rho}{\varrho-2\rho}\}<\theta<1$,
then $\varrho<2\theta^*$, which contradicts the singularity of $\varrho$. Therefore, $\mathscr{L}[I](\lambda)$ is well-defined for all $\text{Re}\lambda>0$.
By \eqref{5.1},
\begin{equation}\label{5.2}
\mathscr{L}[I](\lambda)=\frac{\gamma_1\int_{0}^h f(\tau)\me^{-\lambda c\tau}\dif \tau\int_{-\infty}^{\infty} [K_1-S(z)]I(z)\me^{-\lambda z}\dif z+\gamma_2\int_{-\infty}^{\infty} \big[K_2-P(z)\big]I(z)\me^{-\lambda z}\dif z}{\triangle_K(\lambda,c)}.
\end{equation}
Letting $\lambda\to\infty$ in \eqref{5.2}, it follows from $\triangle_K(\lambda,c)\to\infty$ that the right-hand side of  \eqref{5.2}
is zero. This contradicts the definition of  $\mathscr{L}[I](\lambda)$. Therefore, we have proved the non-existence of epidemic waves.

\end{proof}
{\section{Numerical simulations}}

Choose $K_1=K_2=b_1=b_2=d_1=d_2=d_3=\gamma_1=\gamma_2=1$, $\sigma=0.3$, $\delta=0.8$, $f(\tau)=\cos\tau$, $h=\frac{\pi}{2}$, $N_1=100$.
We obtain the  positive equilibrium  $(S_*,P_*,I_*)=(\frac{43}{85},\frac{5}{17},\frac{12}{5})$ of \eqref{2.1}. From Lemma \ref{L1}(i), we see that $(\lambda^*,c^*)$ is determined by $\triangle_{K}(\lambda^*,c^*)=0$ and $\partial_{\lambda}\triangle_{K}(\lambda^*,c^*)=0$. By using Matlab, it follows that    $\lambda^*\approx 1.0687,c^*\approx 1.3630$.  See figure 2 for more details. Picking $c=\frac{40}{\pi}>c^*$, then we have the nonmonotone epidemic waves of   \eqref{2.1}.  See figures  3-6 for more details. When  $c=c^*$, the figures of epidemic waves is quite similar as the case $c>c^*$. From \eqref{2.3}, we see that the diffusion rate of the infective individuals, conversion rate of the disease can
increase the speed of disease propagation. While the  death rate or removal rate of infected individuals and the delay of disease can weaken
the speed of disease propagation. Moreover, the  form of dispersal kernel $J(\cdot)$  and mode of incubation  $f(\cdot)$ both affect the  speed of disease propagation.

\begin{figure}[htbp]
	\centering
	\includegraphics[height=6.0cm,width=9.5cm]{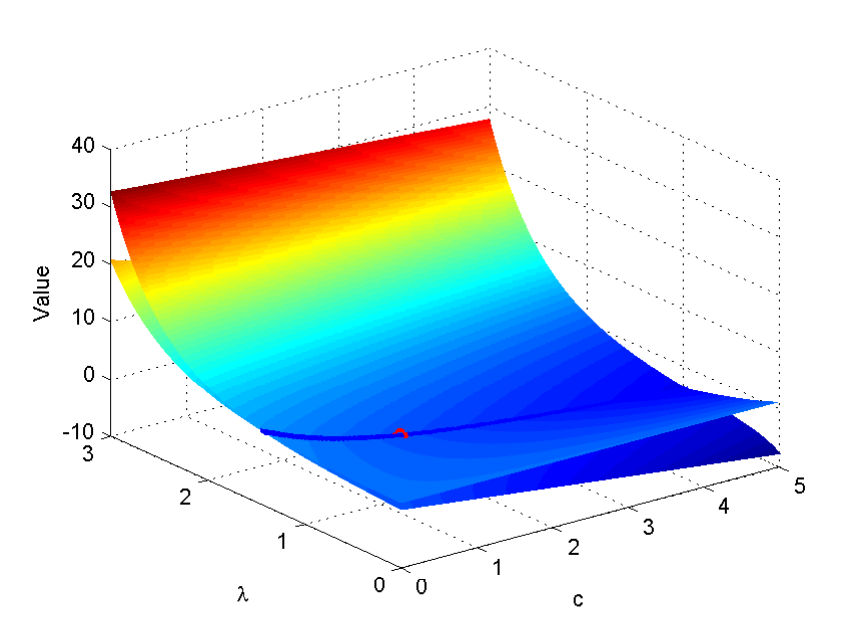}
	\caption{The figure of $(\lambda^*,c^*)$.}
\end{figure}

\begin{figure}[htbp]
	\centering
	\includegraphics[height=8.0cm,width=11.5cm]{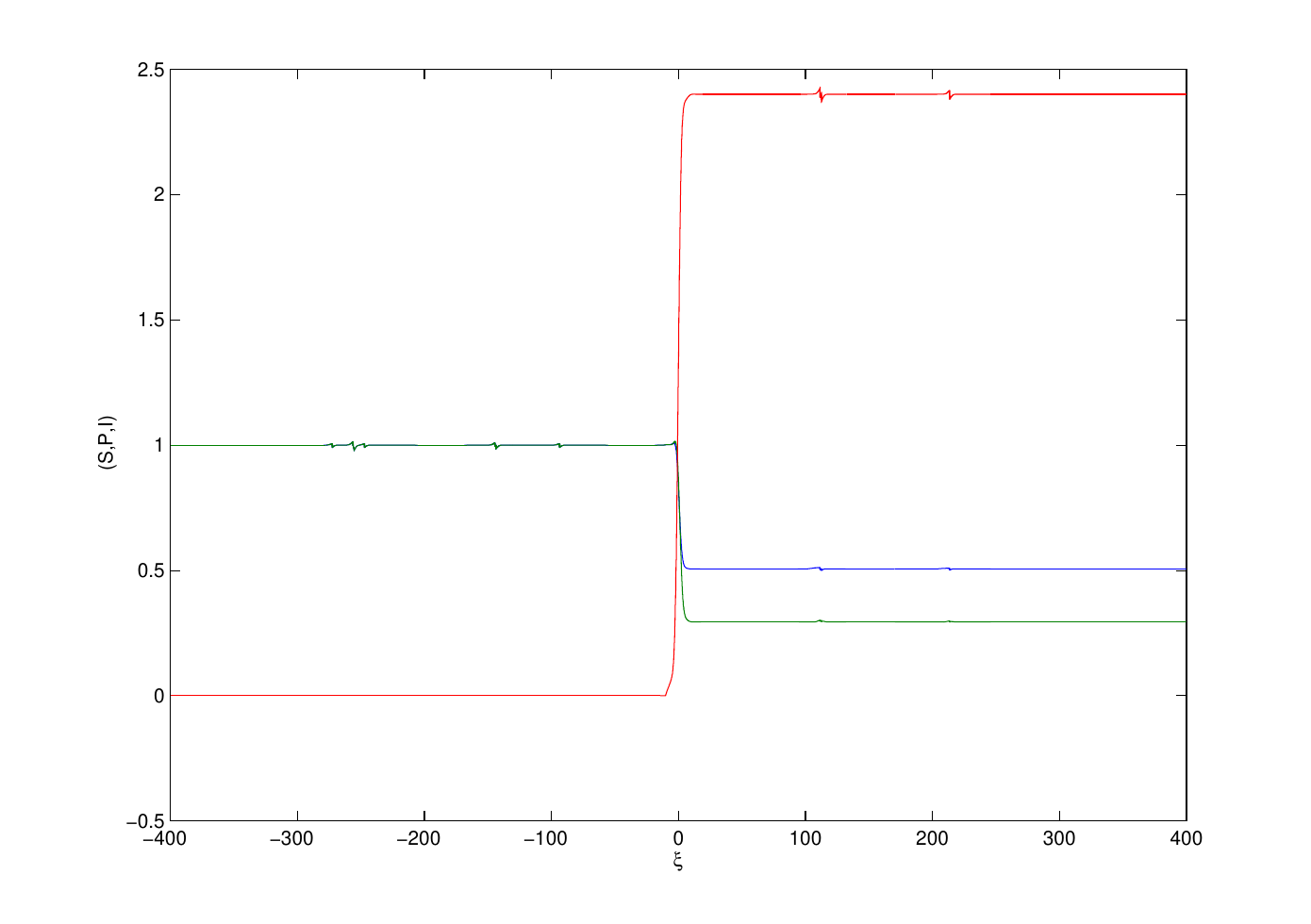}
	\caption{The figure of $(S,P,I)$ for $c>c^*$ with $\xi=j+ct$. The $S$ component is plotted
		by the blue line, $P$ component is plotted by the green line and $I$ component is plotted
		by the red line. }
\end{figure}

\begin{figure}[htbp]
	\centering
	\includegraphics[height=6.0cm,width=9.5cm]{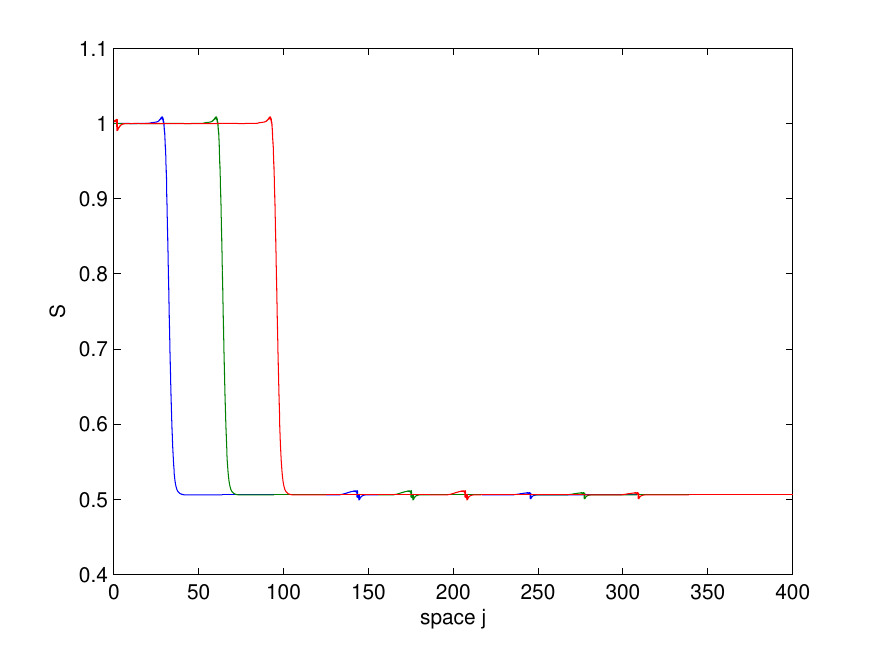}
	\caption{The figure of $S$ for $c>c^*$  is plotted with every $5$ time steps.}
\end{figure}

\begin{figure}[htbp]
	\centering
	\includegraphics[height=6.0cm,width=9.5cm]{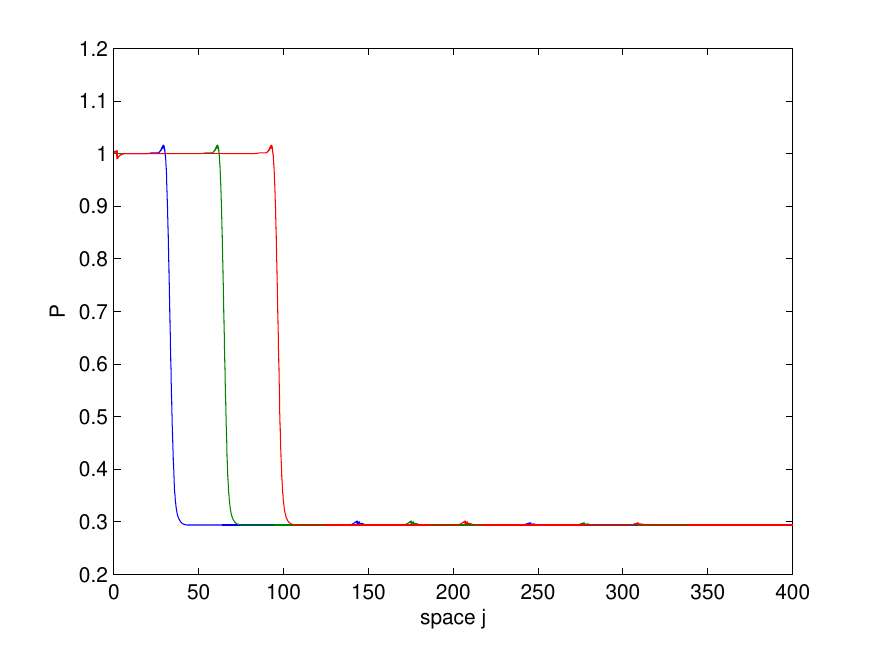}
	\caption{The figure of $P$ for $c>c^*$  is plotted with every $5$ time steps.}
\end{figure}

\begin{figure}[htbp]
	\centering
	\includegraphics[height=6.0cm,width=9.5cm]{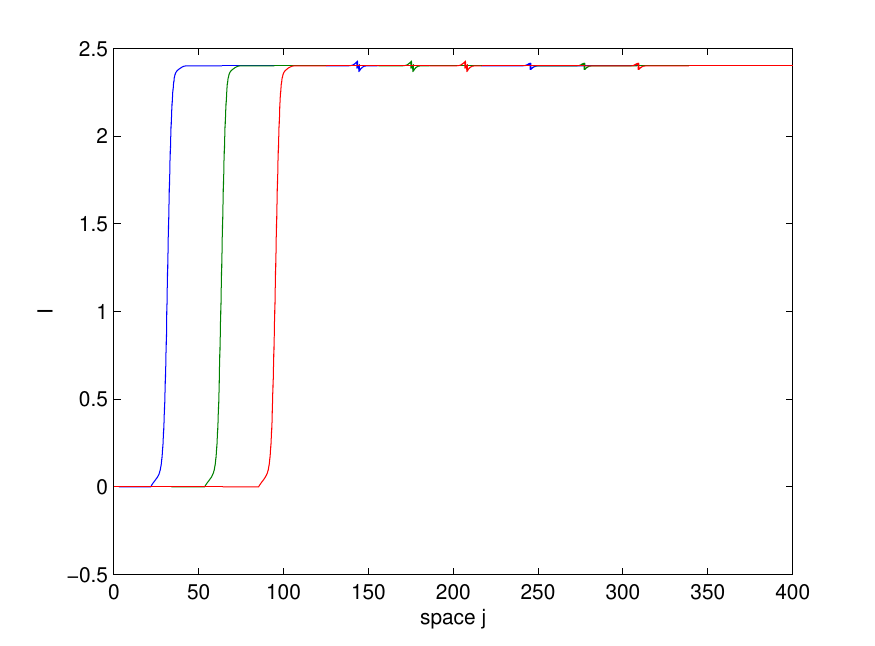}
	\caption{The figure of $I$ for $c>c^*$  is plotted with every $5$ time steps.}
\end{figure}

%




\begin{thebibliography}{100}

\bibitem{am} R. Anderson, R. May,  Infectious diseases of humans: dynamics and control, Oxford University Press,
Oxford,  1991.


\bibitem{av} J. Arino, P. van den Driessche, Disease spread in metapopulations, in: H.
Brunner, X.-Q. Zhao, X. Zou (Eds.), Nonlinear Dynamics and Evolution
Equations, Fields Inst. Commun. 48, Amer. Math. Soc., Providence, RI, 2006, 1-12.

\bibitem{bv} F. Brauer, P. van den Driessche, Models for transmission of disease with immigration of infectives,
Math. Biosci. 171 (2001) 143-154.


\bibitem{bvw} F. Brauer, P. van den Driessche, L. Wang, Oscillations in a patchy environment disease model,
 Math. Biosci. 215 (2008) 1-10.

\bibitem{cy} C. Castillo-Chavez, A.A. Yakubu,  Dispersal, disease and life-history evolution, Math. Biosci. 173 (2001) 35-53.


\bibitem{cmv} J.W. Cahn, J. Mallet-Paret, E.S. van Vleck, Travelling wave solutions for systems of ODE's on a
two-dimensional spatial lattice, SIAM J. Appl. Math. 59 (1999) 455-493.



\bibitem{cg} X. Chen, J.-S. Guo, Existence and asymptotic stability of traveling waves of discrete quasilinear monostable equations, J. Differential
Equations 184 (2002) 549-569.

\bibitem{cfg} X. Chen, S.-C. Fu, J.-S. Guo, Uniqueness and asymptotics of traveling waves of monostable dynamics
on lattices, SIAM J. Math. Anal. 38 (2006) 233-258.



\bibitem{chen} X. Chen,  J.-S. Guo, Uniqueness and existence of traveling waves of discrete quasilinear monostable dynamics, Math. Ann. 326 (2003) 123-146.

\bibitem{chengw} X. Chen, J.-S. Guo, C.-C. Wu, Traveling waves in discrete periodic media for bistable dynamics, Arch.
Ration. Mech. Anal. 189 (2008) 189-236.


\bibitem{cgh}  Y.-Y Chen, J.-S. Guo, F. Hamel,  Traveling waves for a lattice dynamical
system arising in a diffusive endemic model, Nonlinearity 30 (2017) 2334-2359.


\bibitem{cpw}  S.N. Chow, J. Mallet-Paret, W. Shen, Travelling waves in lattice dynamical systems, J. Differential Equations
149 (1998) 249-291.


\bibitem{cf} R. Coutinho, B. Fernandez, Fronts and interfaces in bistable extended mappings, Nonlinearity 11 (1998)
1407-1433.

\bibitem{deb} G. Dwyer, J.S. Elkinton, J. Buonaccorsi, Host heterogeneity in susceptibility and disease dynamics: tests of a mathematical
model, Amer. Nat. 150 (1997) 685-770.



\bibitem{fwz1} J. Fang, J. Wei, X.-Q. Zhao, Uniqueness of traveling waves for nonlocal lattice equations, Proc. Amer.
Math. Soc. 139 (2011) 1361-1373.

\bibitem{fz}  J. Fang, X.-Q. Zhao, Bistable traveling waves for monotone semiflows with applications. J. Eur. Math. Soc. 17 (2015) 2243-2288.


\bibitem{fgw} S.-C. Fu, J.-S. Guo, C.-C. Wu, Traveling wave solutions for a discrete diffusive epidemic model, J. Nonlinear Convex Anal.
17 (2016) 1739-1751.


\bibitem{gh} J.-S. Guo, F. Hamel, Front propagation for discrete periodic monostable equations, Math. Ann. 335
(2006) 489-525.

\bibitem{gw} J.-S. Guo, C.-H. Wu, Traveling wave front for a two-component lattice dynamical system arising in
competition models, J. Differential Equations 252 (2012) 4357-4391.

\bibitem{hz1} D. Hankerson, B. Zinner, Wavefronts for a cooperative tridiagonal system of differential equations, J.
Dyn. Diff. Equations 5 (1993) 359-373.


\bibitem{Hethcote1}  H.W. Hethcote, Qualitative analyses of communicable disease models, Math. Biosci. 28 (1976) 335-356.

\bibitem{Hethcote2}  H.W. Hethcote, The mathematics of infectious diseases, SIAM Rev. 42 (2000) 599-653.


\bibitem{hhv} A. Hoffman, H.J. Hupkes, E.S. Van Vleck, Multi-dimensional stability of waves travelling through rectangular lattices in rational directions, Trans. Amer. Math. Soc. 367 (2015) 8757-8808.

\bibitem{hvw}  Y.-H. Hsieh, P. van den Driessche, L. Wang, Impact of travel between patches for spatial spread of
disease, Bull. Math. Biol. 69 (2007)  1355-1375.


\bibitem{hw}  W. Huang, C. Wu, Non-monotone waves of a stage-structured SLIRM epidemic model with latent period, Proc. Roy. Soc. Edinburgh Sect. A 151 (2021) 1407-1442.


\bibitem{hz} W. Hudson, B. Zinner, Existence of traveling waves for a generalized discrete Fisher's equation, Comm.
Appl. Nonlinear Anal. 1 (1994) 23-46.

\bibitem{k} J.P. Keener, Propagation and its failure in coupled systems of discrete excitable cells, SIAM J. Appl.
Math. 47 (1987) 556-572.

\bibitem{lz} J. Li, X. Zou, Dynamics of an epidemic model with non-local
infections for diseases with latency over a patchy environment, J. Math. Biol. 60 (2010) 645-686.


\bibitem{lzhang} Z. Li, T. Zhang, Boundedness and persistence of traveling wave solutions for a
non-cooperative lattice-diffusion system with time delay, Nonlinear Anal. Real World Appl. 75 (2024) 103968.

\bibitem{lm} A.L. Lloyd, R.M. May,  Spatial heterogeneity in epidemic models, J. Theor. Biol. 179 (1996) 1-11.



\bibitem{mz} S. Ma, X. Zou, Existence, uniqueness and stability of travelling waves in a discrete reaction-diffusion
monostable equation with delay, J. Differential Equations  217 (2005)  54-87.




\bibitem{sw}  X.-F. San, Z.-C. Wang, Traveling waves for a two-group epidemic model with latent
period in a patchy environment, J. Math. Anal. Appl. 475 (2019) 1502-1531.

\bibitem{sh} X.-F. San,  Y. He,  Traveling waves for a two-group epidemic model with latent period and bilinear incidence in a patchy environment, Commun. Pure Appl. Anal. 20 (2021) 3299-3318.


\bibitem{sv} M. Salmani,  P. van den Driessche,  A model for disease transmission in a patchy environment, Discrete
Contin. Dyn. Syst. Ser. B 6 (2006) 185-202.



\bibitem{wm} W. Wang, G. Mulone, Threshold of disease transmission on a patch environment, J. Math. Anal. Appl. 285 (2003) 321-335.

\bibitem{wz} W. Wang, X.-Q. Zhao, An epidemic model in a patchy environment, Math. Biosci. 190 (2004) 97-112.

\bibitem{w}  W. Wang, Epidemic models with population dispersal, in: Y. Takeuchi, Y.
Iwasa, K. Sato (Eds.), Mathematics for Life Sciences and Medicine, Springer,
Berlin, Heidelberg, 2007, 67-95.



\bibitem{wu} C.-C. Wu, Existence of traveling waves with the critical speed for a discrete diffusive epidemic model, J. Differential
Equations 262 (2017) 272-282.





\bibitem{xm} Z.Q. Xu,  J.Y. Ma,  Monotonicity, asymptotics and uniqueness of travelling wave solution of a non-local delayed lattice system, Discrete Contin. Dyn. Syst. 35 (2015) 5107-5131.



\bibitem{zlw} G.B. Zhang, W.-T. Li, Z.-C. Wang, Spreading speeds and traveling waves for nonlocal
dispersal equations with degenerate monostable nonlinearity. J. Differential Equations 252 (2012) 5096-5124.

\bibitem{zwl}  R. Zhang, J. Wang, S. Liu, Traveling wave solutions for a class of discrete diffusive SIR epidemic model, J. Nonlinear
Sci. 31 (2021) 10.

\bibitem{zsw} J. Zhou, L. Song, J. Wei, Mixed types of waves in a discrete diffusive epidemic model with nonlinear incidence and time
delay, J. Differential Equations 268 (2020) 4491-4524.

\bibitem{z} B. Zinner, Existence of traveling wavefront solutions for the discrete Nagumo equation, J. Diff. Equations
96 (1992) 1-27.


\bibitem{zhh}  B. Zinner, G. Harris, W. Hudson, Traveling wavefronts for the discrete Fisher's equation, J. Differential
Equations  105 (1993) 46-62.
\end{thebibliography}

\end{document}